\documentclass[journal]{IEEEtran}
\usepackage{cite}
\usepackage{amsfonts}
\usepackage{bm}
\usepackage{mathtools}
\usepackage{amsmath} 
\usepackage{amssymb}  
\usepackage[ruled]{algorithm2e}
\usepackage{algpseudocode}

\usepackage{amsthm,amsmath,amssymb,amsfonts}
\usepackage{graphicx}
\usepackage{textcomp}
\usepackage{xcolor}
\usepackage{amsmath}
\usepackage{amssymb}
\usepackage{subfigure}
\usepackage{epstopdf}	
\usepackage{amsthm}
\usepackage{array}
\usepackage{multirow}
\usepackage{longtable}
\usepackage{rotating}
\usepackage{booktabs}
\usepackage{nomencl}
\usepackage{float}
\usepackage{mathrsfs}
\usepackage{url}

\usepackage{booktabs,ragged2e}
\usepackage[flushleft]{threeparttable}

\makenomenclature

\theoremstyle{definition}

\newtheorem{remark}{Remark}
\newtheorem{assumption}{Assumption}

\newtheorem{lemma}{Lemma}

\newtheorem{theorem}{Theorem}

\theoremstyle{plain}

\newtheoremstyle{noparens}%
{}
{}%
{\itshape}
{}%
{\bfseries}
{.}%
{ }%
{\thmname{#1}\thmnumber{ #2}\mdseries\thmnote{ #3}}

\theoremstyle{noparens}

\usepackage{fancyhdr}
\usepackage{hyperref}
\fancypagestyle{arxiv}{%
  \fancyhf{} 
  \fancyhead[C]{%
    This is the author’s version of the article accepted for publication in \textit{IEEE Transactions on Power Systems}%
  }
  \fancyfoot[C]{%
    Final version available at \href{https://ieeexplore.ieee.org/abstract/document/10148805}{doi:10.1109/TPWRS.2023.3282368}%
  }
}

\makeatletter
\newcommand{\removelatexerror}{\let\@latex@error\@gobble}
\makeatother

\ifCLASSINFOpdf
\else
\fi
\begin{document}
\title{Optimization-based Ramping Reserve Allocation of BESS for AGC Enhancement}
\author{Yiqiao~Xu,~\IEEEmembership{Student Member,~IEEE,}
        Alessandra~Parisio,~\IEEEmembership{Senior Member,~IEEE,}
        Zhongguo~Li,~\IEEEmembership{Member,~IEEE,}
        Zhen~Dong,~\IEEEmembership{Member,~IEEE,}
        and
        Zhengtao~Ding,~\IEEEmembership{Senior Member,~IEEE}
\thanks{This work was partially supported by the projects Supergen Energy Networks Hub (EP/S00078X/2) and Multi-energy Control of Cyber-Physical Urban Energy Systems (EP/T021969/1). The authors are with the Department of Electrical and Electronic Engineering, The University of Manchester, Manchester M13 9PL, U.K. (e-mail: yiqiao.xu@manchester.ac.uk; zhengtao.ding@manchester.ac.uk).}
\thanks{\noindent
\copyright~2023 IEEE. Personal use of this material is permitted. 
Permission from IEEE must be obtained for all other uses, in any current or future media, 
including reprinting/republishing this material for advertising or promotional purposes, 
creating new collective works, for resale or redistribution to servers or lists, 
or reuse of any copyrighted component of this work in other works.  
The final version of record is available at:  
\href{https://ieeexplore.ieee.org/abstract/document/10148805}{doi:10.1109/TPWRS.2023.3282368}.
}}
\maketitle

\thispagestyle{arxiv}   
\pagestyle{arxiv}       

\begin{abstract}
The transient behavior of Automatic Generation Control (AGC) systems is a critical aspect of power system operation. Therefore, fully extracting the potential of Battery Energy Storage Systems (BESSs) for AGC enhancement is of paramount importance. In light of the challenges posed by diverse resource interconnections and the variability associated, we propose an online optimization scheme that can adapt to changes in an unknown and variable environment. To leverage the synergy between BESSs and Conventional Generators (CGs), we devise a variant of the Area Injection Error (AIE) as a measure to quantify the ramping needs. Based on this measure, we develop a distributed optimization algorithm with adaptive learning rates for the allocation of the ramping reserve. The algorithm restores a larger learning rate for compliance with the ramping needs upon detecting a potentially destabilizing event. We demonstrate the effectiveness and scalability of the proposed scheme through comprehensive case studies. It is shown that the proposed scheme can improve the transient behavior of the AGC system by bridging the gap in ramping capability. 
\end{abstract}

%
\IEEEpeerreviewmaketitle

\section*{List of Key Abbreviations}
\addcontentsline{toc}{section}{List of Main Abbreviations}
\begin{IEEEdescription}[\IEEEusemathlabelsep\IEEEsetlabelwidth{$V_1,V_2,V_3$}]
	\item[AGC] Automatic Generation Control
	\item[ACE] Area Control Error
	\item[AIE] Area Injection Error
	\item[BESS] Battery Energy Storage System
	\item[CG] Conventional Generator
	\item[ISO] Independent System Operator
	\item[FFR] Fast Frequency Reserve
	\item[OCO] Online Convex Optimization
	\item[ORRA] Optimization-based Ramping Reserve\\ Allocation
 	\item[RA] Resource Allocation
	\item[RBF] Radial Basis Function
	\item[SoC] State-of-Charge
	\item[GDB] Governor Dead-Band
	\item[GRC] Generation Rate Constraint
\end{IEEEdescription}
\section*{List of Main Variables and Notations}
\addcontentsline{toc}{section}{List of Variables}
\begin{IEEEdescription}[\IEEEusemathlabelsep\IEEEsetlabelwidth{$V_1,V_2,V_3$}]
    \item[$i,j$] Index for bus and area
    \item[$t,k$] Time index for optimization stage and entire operation span
	\item[$\Delta f_j$] Area frequency deviation
	\item[$ACE$] Area control error
	\item[$AIE$] Area injection error
	\item[$H_j$,$D_j$] Equivalent inertia and damping
	\item[$\Delta P_j^{\rm tie}$] Deviation in tie-line power flows
	\item[$\Delta P_j^{\rm l}$] Deviation in load power
	\item[$\Delta P_i^{\rm m}$] Deviation in CG mechanical power
	\item[$P_i^{\rm b}$] BESS discharge (charge) power
	\item[$\Delta P_i^{\rm gov}$] Deviation in CG governor output
	\item[$\Delta u_i^{\rm gov}$] Deviation in CG governor input
	\item[$u_i^{\rm AGC}$] AGC signal
	\item[$R_i$] Governor droop
	\item[$\mathscr F_i^{\rm GDB}$] Governor dead-band
	\item[$\mathscr F_i^{\rm GRC}$] Generation rate constraint
	\item[$\tau$] Control interval
	\item[$\widehat{AIE}_{i,t}$] Improved AIE perceived by each bus
	\item[$d_{i,t},c_{i,t}$] Power reference signals for discharging and charging of BESS
	\item[$u_{i,t}$] Decision variables to be optimized, where $u_{i,t}=[d_{i,t},-c_{i,t}]^\top$
	\item[$\lambda_{i,t}$] Local Lagrangian multiplier
	\item[$y_{i,t}$] Local information about global constraint
	\item[$f_{i,t}$] Local cost function
	\item[$\kappa_{i,t},\eta_{i,t}$] Adaptive learning rates
\end{IEEEdescription}

\section{Introduction}
\IEEEPARstart{A}{s} countries strive to replace coal-based power generation with renewable energy sources (RESs), power systems are undergoing a transition to support a more diverse range of energy resources. Automatic Generation Control (AGC) is a decentralized balancing mechanism that operates in tens of seconds. In response to net-load variability, local balancing authorities in each area are required to maintain the scheduled system frequency and tie-line power flows while minimizing inter-area oscillations. To achieve this goal, the Area Control Error (ACE) has played an important role \cite{207324}. However, with the growing use of intermittent and stochastic RESs, the regulation burden has become more challenging due to the immature management of these resources \cite{8801901,9399101}. As a result, policymakers and Independent System Operators (ISOs) worldwide are actively exploring the commercial use of Battery Energy Storage Systems (BESSs) to provide grid services.

Several studies \cite{7542188,8447249,8935195} have demonstrated that a reasonably sized Battery Energy Storage System (BESS) can improve Automatic Generation Control (AGC) performance and alleviate pressure on Conventional Generators (CGs). This is due to two factors. Firstly, unlike CGs, BESS features faster dynamics and can better track fast-changing regulation signals. Secondly, BESS can provide symmetric support in both directions and can switch directions instantly. Over the past few decades, many utility-scale BESS projects with AGC functions have been commissioned, and there is a growing trend of coordinating multiple BESSs via a communication network to provide substantial support. For instance, Southern California Edison installed a 10 MW BESS and an 8 MW BESS at different transmission substations \cite{7542188}, while in Germany, an aggregated capacity of 90 MW BESS was equally distributed among six sites \cite{steag}. However, it has been observed in some cases that Battery Energy Storage Systems (BESSs) do not efficiently contribute to the minimization of ACE and can even cause counterproductive regulation. This is partly due to the slow components of CGs that should be self-balanced, but may not be able to be addressed in time under a high penetration rate of BESSs \cite{7115170}. Furthermore, the "neutrality needs" of energy storage may require a portion of BESSs to act in the opposite direction to prevent overcharging or over-discharging \cite{impact}. This can continue to occur if energy-neutral operation is not taken seriously, and BESSs remain involved in AGC after the transients. Additionally, BESSs may overly correct the ACE due to a lack of coordination. This overcompensation creates a regulation requirement in the opposite direction of the area imbalance and can lead to sustained oscillations in system frequency \cite{9445583}.

Academic efforts have been made to fully utilize the potential of BESSs for AGC enhancement, which remains an open challenge. Previous schemes for coordinating utility-scale BESSs and CGs typically followed a priority or capacity-based AGC participation strategy \cite{Cheng:2014,7115170}. To facilitate their participation, existing small-capacity BESSs could be aggregated into a larger entity, which is sometimes referred to as a Virtual Power Plant (VPP). Recent research has investigated the coordinated control of a VPP, which consists of distributed BESSs and heat pump water heaters. For example, \cite{Oshnoei:2022} prioritizes BESSs to respond to ACE beyond the allowable range and identifies the participation factors for VPPs/thermal power plants through multi-objective optimization. In \cite{Oshnoei:2020}, a two-layer Model Predictive Control (MPC) scheme involving distributed BESSs in AGC is proposed, with an ancillary-nominal architecture that provides more efficient control signals to the BESSs, thereby showing superior capability in dealing with uncertainties. To reduce the adverse effects of uncertainty and improve the load-frequency characteristic, \cite{Rafiee:2021} adopts feedback where the BESS contributes slightly to frequency recovery, which is in line with current trends in low-inertia power systems.

Given the difficulty in predicting ACE, decisions may have to be made without future information. To address this, \cite{8449100} tailors an online control policy with a threshold structure for BESS to optimally follow the AGC signal, which is not online as it implements control only after the optimum is approached through a number of iterations. To reduce computational complexity and enable fast online computation, \cite{Stanojev:2020} proposes a decentralized control scheme based on Explicit MPC, which approximates the control laws in an explicit form. Other optimization-based approaches include Approximate Dynamic Programming (ADP) \cite{9721646} and Deep Reinforcement Learning (DRL) \cite{9509287,9310351}. However, DRL needs to be pre-trained with massive data and then deployed online, while ADP can be implemented online but may require extensive computational power, especially when the prediction horizon is large. In contrast, Online Convex Optimization (OCO) requires notably less computational power and is promising for real-time implementation \cite{7044563}. It is an online process requiring agents to repetitively interact with the environment with unknown dynamics for policy improvement. Inspired by recent developments in multi-agent systems, \cite{zhao:2020} combines OCO with a consensus protocol for coordinating multiple BESSs in a fully distributed fashion, which is different from most of the centralized research described above. The algorithm follows the paradigm of Resource Allocation (RA), which brings about two issues that will be elaborated on in one of the research gaps.

From a market practice perspective, energy storage neutrality is crucial for maintaining the operating integrity of BESS, but it has not been taken seriously in many research studies \cite{zhao:2020,Stanojev:2020,Cheng:2014,7115170}. As a consequence, some BESSs may have to operate in opposition to the expected regulation to recover State-of-Charge (SoC) \cite{9103134}, or require a comprehensive SoC control that limits BESS operations to designated periods \cite{7542188}. PJM, an Independent System Operator (ISO) in the US, splits the ACE into a biased signal for slow ramping resources and a hard neutral signal for fast ramping resources like BESS. Since 2017, PJM has switched to "conditional neutrality" because the previous "hard neutrality" was ultimately a poor design from a long-term perspective to minimize ACE \cite{neutrality}. Midcontinent ISO (MISO) has introduced a different market design known as AGC Enhancement to better utilize these ramping reserves, which prioritizes BESSs in AGC and withdraws their deployments in batches once the system frequency is restored \cite{7285889}. Furthermore, there are some practical concerns associated with the use of AGC enhancement and other schemes, which originate from the underlying architecture of AGC, i.e., the calculation ACE. Generally, ACE provides a proxy error signal for the true area imbalance by using a static frequency bias. In practice, bias uncertainty, which refers to the discrepancy between the frequency bias and the area's frequency response characteristics, is not uncommon since the frequency bias is static and updated annually \cite{manual}. One major reason for this is that, owing to their slow dynamics, CGs do not strictly adhere to the governor droop during transients. Miscalculated ACE can result in nuisance activation of BESSs, which undermines their efficient operation \cite{syed2019enhanced}. The deficiencies in quantitatively measuring the true value would be exacerbated when turbine-governor nonlinearities are present. To account for the turbine-governor nonlinearities, the concept of Area Injection Error (AIE) has recently been proposed in \cite{9216012}. AIE corrects the ACE to some extent using direct measurements of generator power injections.

To conclude, the following research gaps are observed in the previous studies:
\begin{itemize}
    \item The ACE assumes a static frequency bias and is susceptible to bias uncertainties in practice. The AIE proposed in \cite{9216012} is likely to be closer to the true imbalance and can reduce inter-area oscillations. However, AIE is derived based on a quasi-steady-state approximation and is subject to slow turbine-governor dynamics, resulting in slower convergence than the ACE. Additionally, bias uncertainty on the load side has been overlooked.
	
	\item Most existing control schemes are centralized, with all computations relying on a central controller. However, these methods are cost-inefficient when there are a large number of BESSs to coordinate in real-time due to the high computational and communication requirements on the central controller. More importantly, centralized schemes are vulnerable to a single point of failure and require a powerful computing unit.
	
	\item It is of research interest to investigate optimization in an unknown and highly variable environment. However, the algorithms in \cite{zhao:2020,Zhou:2020} have to adopt a constant learning rate to accommodate their use for ORA. This setting would not suit for ramping reserve allocation and its performance may not have been fully exploited. Unlike \cite{Oshnoei:2020} where BESSs act during load transients only, \cite{zhao:2020} treats them like CGs and would lead to continuous charging/discharging, which could deteriorate their operating integrity.
\end{itemize}

Compared to the previous work, the main contributions and highlights of this paper are summarized as follows:
\begin{itemize}
	\item Inheriting the paradigm of the ACE and utilizing the concept of AIE, we propose a variant of the AIE that removes the quasi-steady-state approximation and introduces a feedback loop to account for the instantaneous bias uncertainty on the generation side. In addition, we incorporate a black-box model, an online interpolated Radial Basis Function (RBF) network \cite{surrogate}, to emulate the bias uncertainty on the load side partially. The proposed AIE exhibits a faster dynamic response than ACE while effectively handling turbine-governor nonlinearities, thereby improving AGC performance even without BESS participation.
	
    \item  By determining AIE as the ramping reserve to be allocated, we propose a novel scheme called ORRA that integrates online learning with OCO, unleashing further the online features for AGC enhancement. During transients, the BESSs act as a complement to the CGs, exploiting their synergistic effect and improving the system's transient behavior, which is a key feature of ORRA. Negligible SoC variations in the long run, approaching energy-neutral operation.
	
    \item A distributed OCO algorithm is developed for ORRA, where a dual-bounded technique \cite{Cao:2021} is integrated to improve compliance with the fast-changing ramping needs. Adaptive learning rates that vary with time, with a two-phase switch mechanism, are developed to cater for both the control and optimization aspects. We prove that, under mild conditions, the algorithm provides guarantees for sublinear dynamic regret and dynamic fit without the use of future or global information that can imply impracticality. Case studies have shown the effectiveness of ORRA in terms of AGC enhancement.
\end{itemize}	

The paper is organized as follows. Section \uppercase\expandafter{\romannumeral2} provides some preliminaries on the interconnected power system. Section \uppercase\expandafter{\romannumeral3} discusses the fundamentals of AGC and the design of the AIE signal. Section \uppercase\expandafter{\romannumeral4} formulates the problem of ramping reserve allocation. In Section \uppercase\expandafter{\romannumeral5}, we present the proposed scheme, the optimization algorithm, and key theoretical results. Comprehensive case studies are presented in Section \uppercase\expandafter{\romannumeral6} to verify its effectiveness through simulations. Finally, we conclude this paper in Section \uppercase\expandafter{\romannumeral7}.

\section{Preliminaries}
\begin{figure*}[h]
	\centering
	\includegraphics[width=1.6\columnwidth]{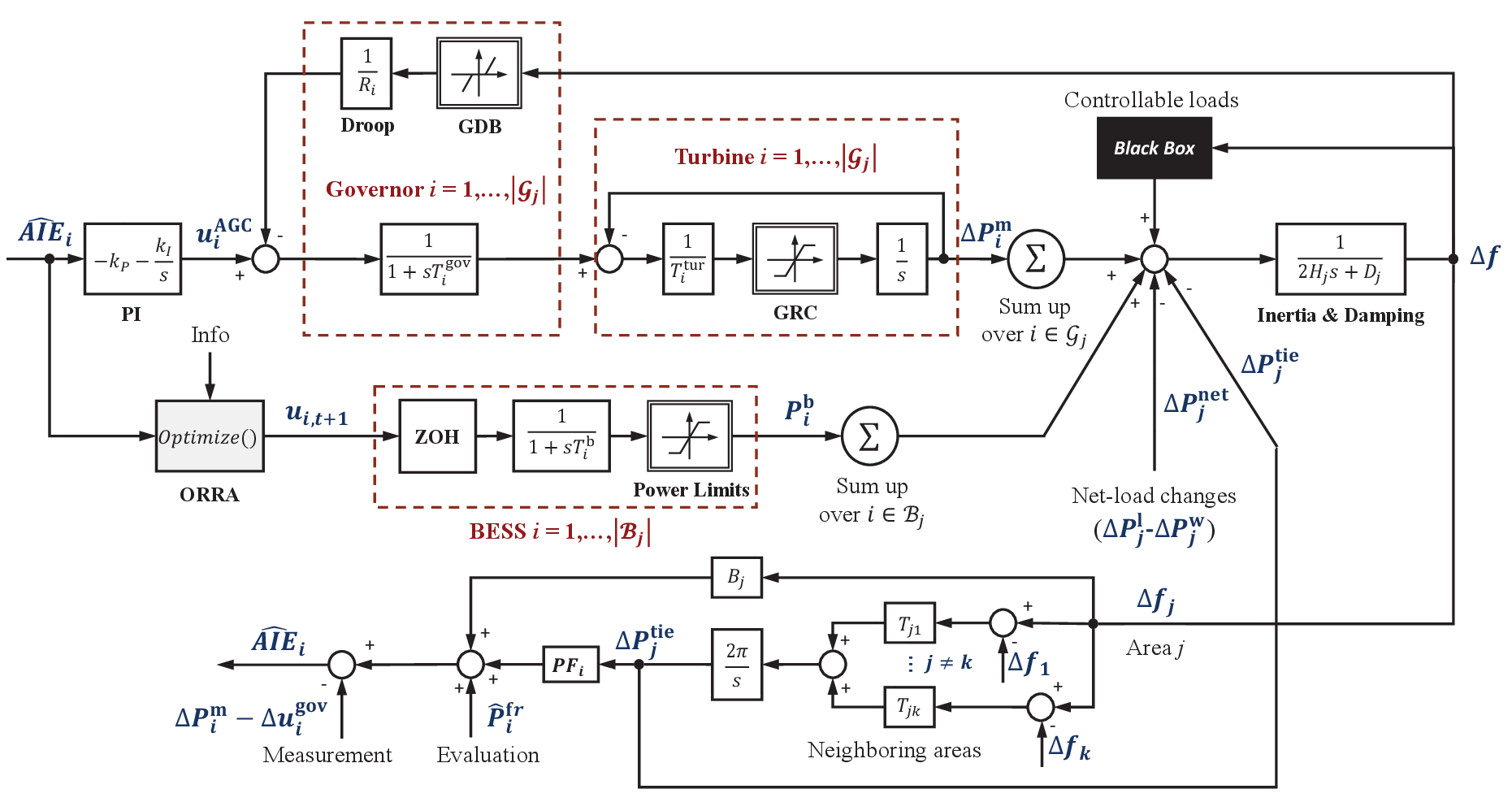}
	\caption{Implementation of AGC in a control area containing multiple CGs and multiple BESSs.}
\end{figure*}
To present the model of an interconnected power system, this section begins with a traditional system that does not account for the penetration of other resources. The interconnected power system is partitioned into multiple control areas, with the set of generator buses denoted by $\mathcal G_j$. A lumped expression for area $j$ can then be obtained \cite{Patel:2019,Rakhshani:2017}
\begin{align}\label{freqency_1}
    \Delta\dot f_j &= \frac{1}{2H_j}\left(\sum_{i\in\mathcal G_j}\Delta P_i^{\rm m} - \Delta P_j^{\rm l} - \Delta P_j^{\rm tie}\right)-\frac{D_j}{2H_j}\Delta f_j,
\end{align}
where $\Delta f_j$ is the frequency deviation of area $j$, $H_j$ is the equivalent system inertia, and $D_j$ is the equivalent system damping \cite{Pathak:2017}. $\Delta P_i^{\rm m}$ is the deviation in mechanical power of the CG connected at bus $i$ relative to an optimizer of economic dispatch. $\Delta P_j^{\rm l}$ is the deviation in loads. An area has either an import or export of power and is tightly coupled with adjacent areas via tie-line power flows $\Delta P_j^{\rm tie}$. The tie-line power flows from area $j$ to its neighboring areas can be presented as follows \cite{Rakhshani:2017}:
\begin{align}
	\Delta\dot P_j^{\rm tie} = \sum_{k\in\mathcal A-\{j\}} 2\pi T_{jk}\left(\Delta f_j-\Delta f_k\right).
\end{align}
where $T_{jk}$ denotes the synchronizing torque between area $j$ and $k$ and $\mathcal A$ is the set of control areas ($\vert\mathcal A\vert\geq 2$).

In this paper, a reduced-order model \cite{Oshnoei:2020} is adopted for the analyzed CG, which consists of a speed governor and a non-reheat steam turbine. In the presence of non-negligible nonlinearities, the turbine-governor response of the CG connected at bus $i\in\mathcal G_j$ can be described by:
\begin{align}
    \Delta\dot P_i^{\rm m} &= \mathscr F_i^{\rm GRC}\left(\frac{\Delta P_i^{\rm gov}-\Delta P_i^{\rm m}}{T_i^{\rm t}}\right),\label{turbine}\\
    \Delta\dot P_i^{\rm gov} &= -\frac{\Delta P_i^{\rm m}}{T_i^{\rm g}} + \frac{1}{T_i^{\rm g}}\left( u_i^{\rm AGC}-\frac{\mathscr F_i^{\rm GDB}(\Delta f_j)}{R_i}\right).\label{governor}
\end{align}
In (\ref{turbine}) and (\ref{governor}), $\mathscr F_i^{\rm GDB}$ and $\mathscr F_i^{\rm GRC}$ describe the Governor Dead-Band (GDB) and Generation Rate Constraint (GRC), which impose a non-negligible nonlinear behavior under particular conditions. $\Delta P_i^{\rm g}$ is the deviation in governor output, $T_i^{\rm t}$ and $T_i^{\rm g}$ are time constants for the turbine and governor, $1/R_i$ is the droop rate for governor speed control, and $u_i^{\rm AGC}$ is the AGC signal generated by passing the ACE or AIE through a PI controller.

\begin{figure}[htbp]
	\centering
	\includegraphics[width=0.7\columnwidth]{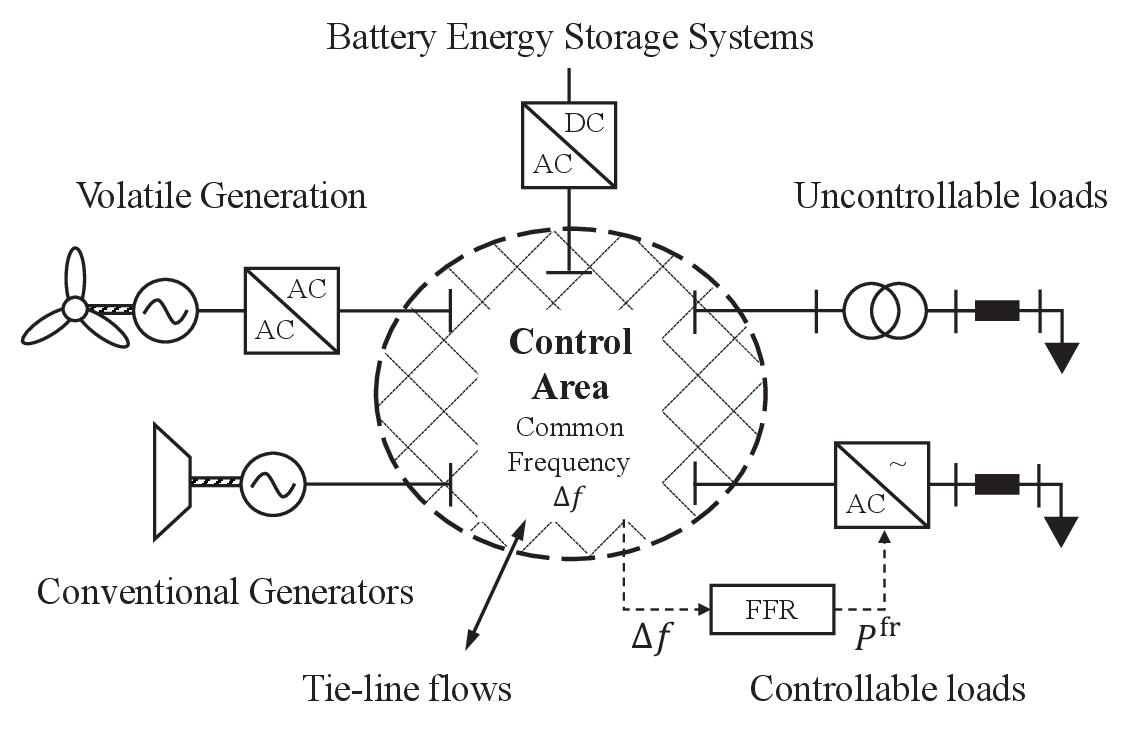}
	\caption{Basic frame of a control area with diverse resource interconnection.}
\end{figure}

To ensure the proper functioning of the power system with an increasing share of RESs, Fast Frequency Reserve (FFR) based on under-frequency load shedding, energy storage \cite{8836644,Shadabi:2022}, direct load control \cite{Zhu:2018,Carne:2019}, etc. have been proposed in the literature, where common control measures include droop \cite{8836644,Carne:2019} (sometimes with dead-zone), sectional droop \cite{Zhu:2018,Datta:2019}, and nonlinear droop \cite{Nutkani:2015,Shadabi:2022}. Given the massive number and diverse composition of frequency-responsive resources, the FFRs $P_i^{\rm fr}$ provided at each bus will be a high level of aggregation
\begin{align}
    P_i^{\rm fr} = \underbrace{K_i^{{P}\text{-}{f}}(\Delta f_j)}_{agg}\Delta f_j,
\end{align}
such that the resulting damping factor $K_i^{{P}\text{-}{f}}$ approaches a sectional droop curve with a large number of segments plus an additional nonlinear function exhibiting different levels of frequency sensitivity. As shown in Fig. 1, it is treated as a ``black-box" in the following analysis.

Considering also the contribution of BESSs in AGC, (\ref{freqency_1}) can be re-arranged as follows to adapt a more general case as shown in Fig. 2:
\begin{align}\label{frequency_2}
    \Delta\dot f_j &= \frac{1}{2H_j}\left(\sum_{i\in\mathcal G_j}\Delta P_i^{\rm m} + \sum_{i\in\mathcal B_j} P_i^{\rm b} - \Delta P_j^{\rm net} - \Delta P_j^{\rm tie}\right)\notag\\
    &\quad-\frac{D_j+\sum_{i\in\mathcal B_j}K_i^{{P}\text{-}{f}}}{2H_j}\Delta f_j,
\end{align}
with the synthetic inertia and damping of volatile generation included in $H_j$ and $D_j$. We replace $\Delta P_j^{\rm l}$ with a new term $\Delta P_j^{\rm net}$ to represent the net-load variation, i.e., variation in load minus variation in volatile generation. Moreover, $P_i^{\rm b}$ represents the instantaneous power output of the BESS connected at bus $i$. It is derived by passing the reference signal through a zero-order holder and a first-order transfer function with output saturation.

\section{Improved AIE for AGC}
When a load perturbation takes place, the CGs and BESSs are obliged to respond to the ACE. The ACE is obtained as the difference between scheduled and actual tie-line power flows $\Delta P_j^\textrm{tie}$ plus a scaled frequency deviation $\Delta f_j$. The ACE for an area is given by
\begin{align}\label{ACE}
	ACE_j &= \Delta P_j^\textrm{tie} + B_j\Delta f_j,
\end{align}
where $B$ represents the frequency bias. It is set that $B_j = D_j + R_j^{-1}$ to make the numerical value of the ACE physically meaningful. Such a static bias setting rests on several assumptions taking place only under ideal conditions. To facilitate our analysis, we denote $\mathcal B_j$ as the set of buses and $\mathcal G_j$ the set of buses with generator.

\subsection{Bias Uncertainty from Turbine-Governor Nonlinearities}
Generally, ACE provides a proxy error signal for the true area imbalance. However, its deficiencies in quantitatively measuring the true value are further evident when turbine-governor nonlinearities are present. As indicated in equations (\ref{turbine}) and (\ref{governor}), explicitly obtaining an analytical expression for the dynamic response of turbine-governor systems can be very difficult. To this end, the AIE emerged as a new concept to account for the bias uncertainty from turbine-governor nonlinearities \cite{9216012}. For ease of analysis, we denote the governor input in incremental form
\begin{align}\label{u_gov}
    \Delta u_i^{\rm gov}\coloneqq u_i^{\rm AGC} - \frac{\mathscr F_i^{\rm GDB}(\Delta f_j)}{R_i}.
\end{align}
By subtracting $\Delta P_i^m$ from both sides of (\ref{u_gov}), we have
\begin{align}\label{quasi}
	\Delta u_i^{\rm gov}-\Delta P_i^{\rm m} &= u_i^{\rm AGC}-\Delta P_i^{\rm m} - \frac{\mathscr F_i^{\rm GDB}(\Delta f_j)}{R_i}.
\end{align}
At quasi-steady-state, $\Delta P_i^{\rm m}\approx\Delta u_i^{\rm gov}$ and hence
\begin{align}
	u_i^{\rm AGC}-\Delta P_i^{\rm m} &\approx \frac{\mathscr F_i^{\rm GDB}(\Delta f_j)}{R_i},
\end{align}
which provides an approach to remove the GDB from signal calculation using direct measurement of $\Delta P_i^\textrm{m}$. Thus, the AIE in \cite{9216012} is constructed as
\begin{align}\label{AIE}
	AIE_j=\Delta P_j^\textrm{tie} + D_j\Delta f_j + \sum_{i\in\mathcal G_j}\left(u_i^{\rm AGC}-\Delta P_i^\textrm{m}\right).
\end{align}
While the AIE has been shown to provide a more accurate estimate of the true area imbalance and reduce inter-area oscillations compared to ACE, its use of the quasi-steady-state approximation can result in slower convergence rates than ACE-based AGC. To address this issue, we propose a modified AIE algorithm that removes the quasi-steady-state approximation while retaining the structure of traditional ACE and incorporating the idea of AIE:
\begin{align}\label{improved_AIE}
	AIE_j &=\sum_{i\in\mathcal G_j} \left(\frac{1}{R_i}-\frac{\Delta P_i^{\rm m}-\Delta u_i^{\rm gov}}{\Delta f_j}\right)\Delta f\notag\\
	&\quad+\Delta P_j^\textrm{tie} + D_j\Delta f_j\notag\\
	&= ACE_j - \sum_{i\in\mathcal G}\left(\Delta P_i^\textrm{m}-\Delta u_i^{\rm gov}\right),
\end{align}
Considering a participating factor $\sigma_i$, for $i\in\mathcal G$ we have
\begin{align}
	AIE_i = \sigma_i ACE_j - \left(\Delta P_i^\textrm{m}-\Delta u_i^{\rm gov}\right),
\end{align}
where $\sum_{i\in\mathcal G}\sigma_i = 1$, while $(\Delta P_i^{\rm m}-\Delta u_i^{\rm gov})/\Delta f_j$ quantifies the instantaneous bias uncertainty that arises during transients, which implicitly leads to a dynamic bias setting. It is worth noting that $AIE_i=0$ for bus without BESS.

\begin{remark}
GDBs are generally classified as either unintentional or intentional. Unintentional GDBs are a result of the inherent mechanical effects of turbine-governor systems, such as sticky valves or loose gears. On the other hand, intentional GDBs are deliberately introduced in governor droop designs to reduce excessive regulation efforts and mechanical wear. In this paper, we consider an intentional GDB of 36mHz, making $\Delta u_i^{\rm gov}$ readily accessible. However, we note that the AIE design can also be extended to handle unintentional GDBs by linearizing the corresponding transfer functions in Fourier space \cite{Khezri:2019}, which allows for an estimation of $\Delta u_i^{\rm gov}$.
\end{remark}

\subsection{Bias Uncertainty from Fast Frequency Response}
Subsequent to the increasing deployment of FFRs \cite{Shadabi:2022, Zhu:2018, Carne:2019}, additional load damping has been introduced into existing power systems, resulting in bias uncertainty on the load side \cite{7882719}. However, quantifying this additional load damping is challenging, as its estimated value may only be valid for the frequency condition for which it was derived. Merely considering a static $B_j$ is insufficient to reflect $\sum_{i\in\mathcal B_j}K_{i}^{{P}\text{-}{f}}$ that can vary considerably with the clearance of intraday markets \cite{8836644}. To address this issue, it would be beneficial to incorporate also the frequency responsive $\sum_{i\in\mathcal B_j}P_i^{\rm fr}$ when calculating the AIE. We would like to clarify that our focus is on evaluating the aggregated $P$-$f$ characteristics of FFRs, and the frequency responses of wind turbines are considered as unpredictable bias uncertainty.

To avoid the extra expenses associated with real-time monitoring of $P_i^{\rm fr}$ and to make it more generalizable, we propose to use online interpolated Radial Basis Function (RBF) networks \cite{surrogate} to produce a local approximation. The aim is not to precisely model but rather to emulate $P_i^\textrm{fr}$ based on a limited number of evaluations. We assume that the $P$-$f$ characteristics remain unchanged until the next market clearing, and the learning process is restarted for each intraday market interval. For bus agent $i$, we denote the datasets by $\mathcal S_i^{\rm f}\coloneqq[\Delta f_{i,1},...,\Delta f_{i,M}]^\top$ and $\mathcal S_i^{\rm P}\coloneqq[\Delta P_{i,1}^\textrm{fr},...,\Delta P_{i,M}^\textrm{fr}]^\top$. During real-time operation, these datasets are gradually expanded by latest information if certain conditions are met. Thereby, we make the following improvement:
\begin{align}
	\widehat{AIE}_i =  AIE_i + \sum_{m=1}^M\omega_{i,m}\phi(\Vert\Delta f_j-\Delta f_{i,m}\Vert),
\end{align}
where the second term is the RBF interpolant, $\omega_{i,m}$ is a weighting factor that needs to be determined individually for each neuron, and $\phi(x)$ is the Gaussian basis function
\begin{align}
	\phi(x) = \exp(-\xi x^2), \quad\xi\in\mathbb R_{>0}.
\end{align}

\begin{figure}[htbp]
	\centering
	\includegraphics[width=0.98\columnwidth]{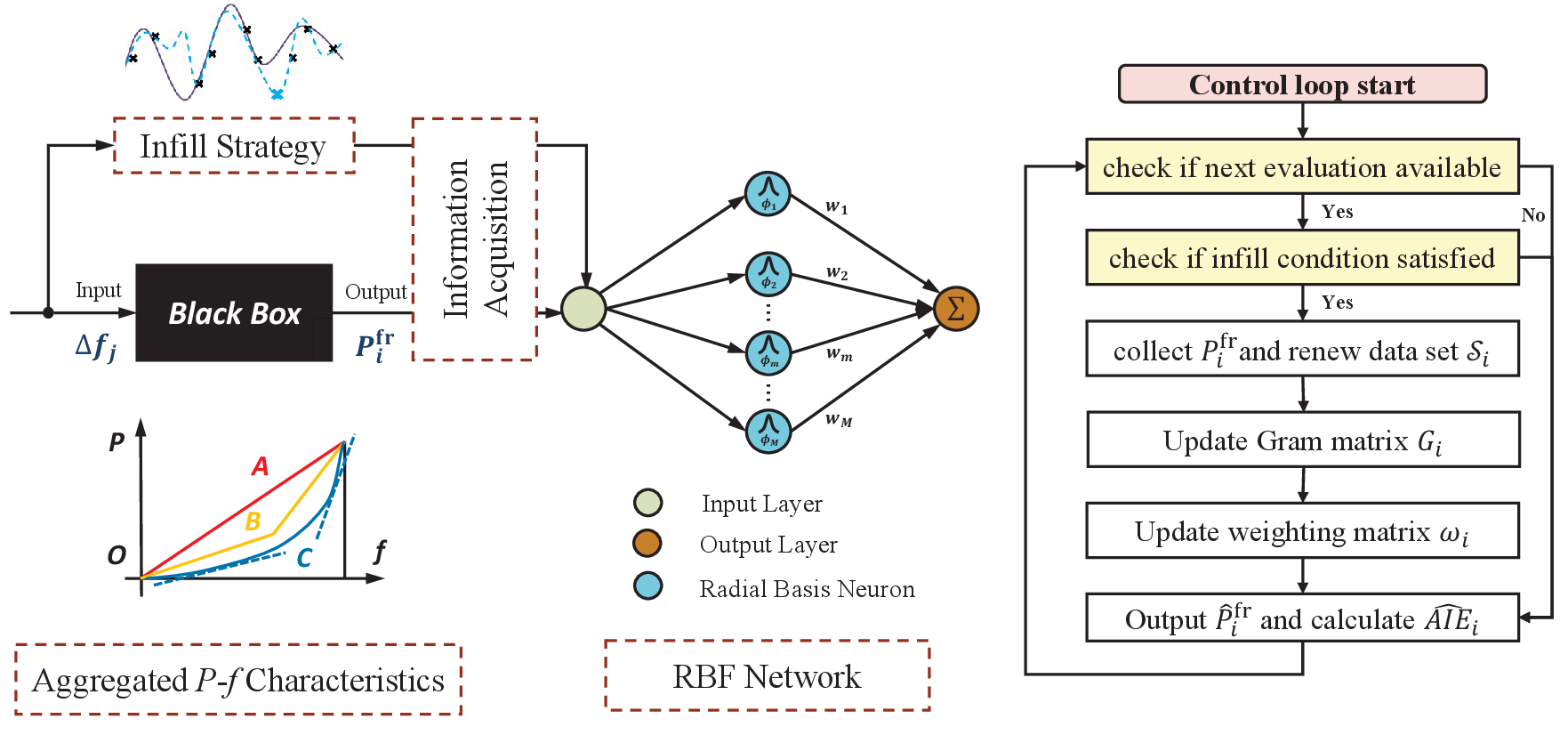}
	\caption{Structure of the online interpolated RBF network.}
\end{figure}

As illustrated in Fig. 3, at the start of a new control interval ($\tau = 0.1s$), each bus agent identifies whether the next evaluation of $\Delta P_i^{fr}$ should be conducted. If so, the local FFRs will be collected by the bus agent along with the current area frequency measurement. To allow for sufficient time for data collection, two contiguous evaluations are kept for at least 5 seconds from each other. A distance-based infill method is adopted from our previous work \cite{surrogate} to determine evaluation points for model improvement. The idea is to ensure that the next evaluation point is held at an adequate distance from the previously evaluated points: 
\begin{align}
    \Vert\Delta f_{j}-\Delta f_{i,m}\Vert\geq \epsilon_{i,M}\mathcal D_{i,\max},\quad\forall m=1,\dots,M,
\end{align}
where $\mathcal D_{i,\max}$ represents the maximal distance that can be reached for the next evaluation, and $\epsilon_{i,M}$ is a coefficient for balancing between exploration and exploitation. The reader is referred to \cite{surrogate} for more details.

The interpolation matrix, also referred as Gram matrix, is updated according to
\begin{align}
	[G_i]_{rc} = \phi(\Vert\Delta f_{i,r}-\Delta f_{i,c}\Vert),\quad\forall r,c=1,...,M,
\end{align}
and the weighting matrix, denoted by $\omega_i=[\omega_{i,1},...,\omega_{i,M}]^\top$, is determined according to
\begin{align}
	\omega_i = (G_i^\top)^{-1}\mathcal S_i^{\rm P}.
\end{align}
There always exists a unique $\omega_i$ such that the RBF interpolant can reproduce observed behaviors.

\section{Problem Formulation of Ramping Reserve Allocation}
The variability and uncertainty of RESs introduce significant challenges to existing power systems. The occurrence of ramp capability shortages in AGC has put a demand on the implementation of fast-ramping reserves. This section describes the ramping reserve allocation problem, where the AIE signals are adopted as the ramping reserve to be allocated among BESSs, as depicted in Fig. 4.
\begin{figure}[htbp]
	\centering
	\includegraphics[width=0.8\columnwidth]{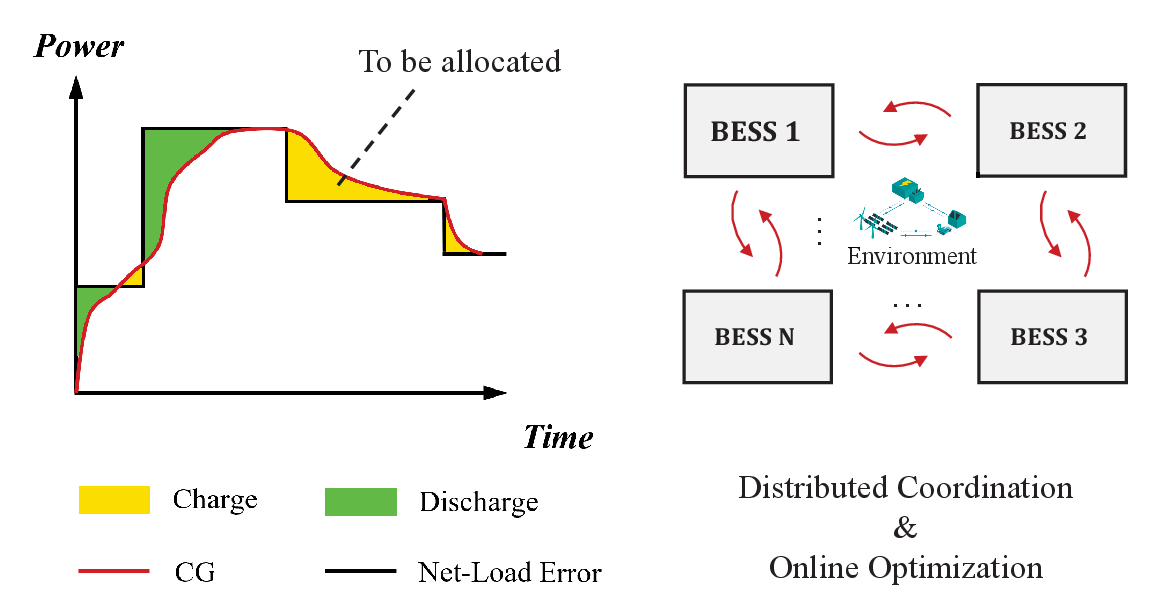}
	\caption{Schematic overview of ramping reserve allocation.}
        \label{overview_allocation}
\end{figure}

\subsection{BESS Model}
Consider a battery operation defined over discrete time, where each control interval has a duration of $\tau$. For $i\in\mathcal B_j$, i.e., the BESS connected at bus $i$, its SoC at the next time instant $k+1$ can be described using a linear difference equation:
\begin{align}\label{soc}
	x_{i,k+1} = x_{i,k} + \frac{\eta^\textrm c\tau}{E_i} c_{i,k+1} - \frac{\tau}{\eta^\textrm d E_i} d_{i,k+1},
\end{align}
where $x_{i,k+1}$ and $x_{i,k}$ are the SoC levels of agent $i$ at time instant $k$ and $k+1$, respectively; $\eta^\textrm c$ and $\eta^\textrm d$ are the charging/discharging efficiencies; $E_i$ is the rated capacity; $c_{i,k+1}$ and $d_{i,k+1}$ denote the reference signals for charging and discharging and are treated as equivalent to the instantaneous BESS powers in this formulation, provided that the internal control loops are fast enough. 

Each BESS can either operate in charging or discharging mode. Irrespective of the model used, one has to avoid simultaneous charging and discharging for efficiency considerations. We introduce a binary variable $\delta_{i,k}$, whose value at time instant $k$ is assigned according to
\begin{align}
	\delta_{i,k} = \left\{
	\begin{aligned}
		& \left(\widehat{AIE}_{i,k}/\vert\widehat{AIE}_{i,k}\vert+1\right)/2,\quad i\in\mathcal G_j,\\
		&\delta_{l,k-\textrm{dist}(i,l)},\quad i\in\mathcal B_j-\mathcal G_j,
	\end{aligned}
	\right.
\end{align}
where relay communication is considered for passing this binary variable to BESSs without access to the AIE signal, and $\textrm{dist}(i,l)$ describes the communication delay between bus $i$ and bus $l$ where $l\in\mathcal G_j$ exhibits the shortest path to $i$. Then, $\delta_{i,k}$ determined at time instant $k$ is used to decide whether charge or discharge at the next time instant $k+1$ by setting bounds on $c_{i,k+1}$ and $d_{i,k+1}$:
\begin{align}
	& 0\leq c_{i,k+1}\leq\left(1-\delta_{i,k}\right)\overline{c_i},\\
	& 0\leq d_{i,k+1}\leq\delta_{i,k}\overline{d_i},
\end{align}
such that the BESS is charged if $\delta_{i,k}=0$ and discharged if $\delta_{i,k}=1$, where $\overline {c_i}$ and $\overline {d_i}$ denote the BESS power limits.

To avoid over-charging and over-discharging, the SoC of each agent needs to be restricted within an appropriate range:
\begin{align}
	\underline{x_i}\leq x_{i,k}+\frac{\eta^\textrm c\tau}{E_i} c_{i,k+1} - \frac{\tau}{\eta^\textrm d E_i} d_{i,k+1}\leq\overline{x_i},
\end{align}
where $\underline{x_i}$ and $\overline{x_i}$ are the minimal and maximal SoC levels.

\subsection{Cost Model}
Cycling aging refers to a natural process leading to permanent battery degradation and is related to the depth for which a battery is cycled. The resultant cost of cycling aging is usually omitted \cite{8454327,Oshnoei:2020,8935195} or approximated through a simplified model \cite{zhao:2020,SUN2020115589,8805438}. We adopt a semi-empirical model that combines cycle identification results with experimental data \cite{8449100}.
Using the well-known rainflow-counting algorithm (due to space limits, please refer to \cite{AMZALLAG1994287}), we can identify the cycle depth of the latest half-cycle per iteration
\begin{align}
	(\mu_{i,k},\mathcal R_{i,k+1}) = \text{Rainflow}(x_{i,k},\mathcal R_{i,k}),
\end{align}
where $\mu_{i,k}$ is the cycle depth between the last two residues, $\mathcal R_{i,k+1}$ is the updated set of residues (the extrema unremoved by the rainflow-counting algorithm), and $x_{i,k}$ is the latest SoC information, which together with $\mathcal R_{i,k}$ actually converts SoC trajectories that entail non-uniform fluctuations into consecutive cycles that can be full or half. A full cycle consists of a charge half-cycle and a discharge half-cycle, and it might be nested within other cycles once new SoC samples are acquired.

Subsequently, we are able to characterize the battery lifetime loss with respect to the identified half-cycle as
\begin{align}
	\Delta L_{i,k}(d_{i,k},c_{i,k})\coloneqq\frac{n_{i,k}^\textrm{cyc}}{2}a \mu_{i,k}^b,
\end{align}
where $a$ and $b$ are empirical coefficients that normalize the cycling aging for a full cycle between 0 and 1, while $n_{i,k}^\textrm{cyc}\in(0,1]$ calculates the number of cycles from the time indexes of the latest two residues. Note that $\Delta L_{i,k}$ is a convex function of $\mu_{i,k}$ and, by the chain rule, also a convex function of $d_{i,k}$ and $c_{i,k}$ \cite{8449100}. Additional quadratic terms on the BESS powers quantify the power wear \cite{8274099}. As a result, the battery usage cost (\$/h) is given as 
\begin{align}\label{cost_function}
\begin{split}
    f_{i,k}(d_{i,k},c_{i,k})&\coloneqq\underbrace{\theta_i^\textrm{a}\cdot(3600/\tau)\cdot\Delta L_{i,k}(d_{i,k},c_{i,k})}_\textrm{Cycling aging cost}\\
    &\quad + \underbrace{\theta_i^\textrm{b}\cdot(d_{i,k}-c_{i,k})^2}_\textrm{Power wear cost}.
\end{split}
\end{align}
The calendar aging independent of charge-discharge cycling is omitted as it is beyond the time frames of ORRA.

\subsection{Optimization Problem Formulation}
Consider $N$ BESSs that are installed across the area. Each of them is managed by the local bus agent, which cannot reveal its cost function to the others. As illustrated in Fig. 5, the entire operation from time instant $0$ to time instant $k+1$ can be divided into a number of optimization stages separated by the reset of learning rates, which will be covered in \uppercase\expandafter{\romannumeral5}. A. Meanwhile, we use $t$ as an index for the current optimization stage, with iteration $0$ denoting its beginning and iteration $t+1$ denoting the current position of optimization. For instance, $(\cdot)_{i,t}$ in the current optimization stage is treated as equivalent to $(\cdot)_{i,k}$ in the entire operation.
\begin{figure}[htbp]
	\centering
	\includegraphics[width=0.98\columnwidth]{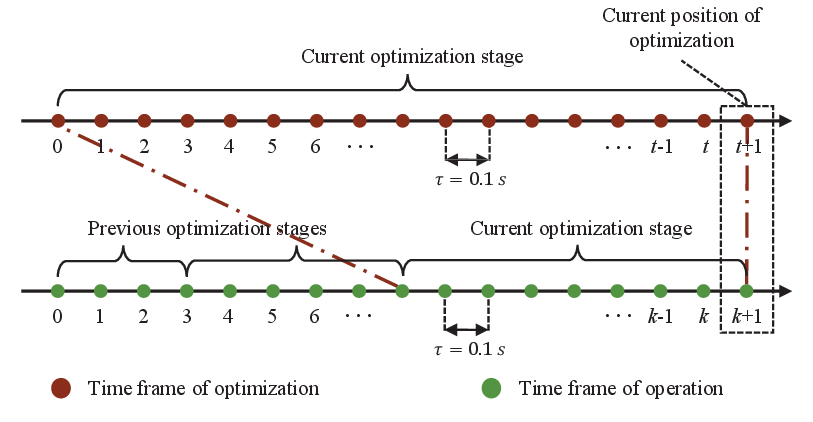}
	\caption{Time frames of optimization and entire operation of BESS. Iteration $t+1$ at current optimization stage corresponds to time instant $k+1$ for the operation span.}
\end{figure}

In terms of cost minimization, the overall optimization problem can be mathematically modeled as follows, which is convex with time-varying constraints
\begin{subequations}\label{problem}
\begin{align}
	&\min_{d_{i,t+1},c_{i,t+1}} \sum_{i=1}^{N}f_{i,t+1}(d_{i,t+1},c_{i,t+1})\\
	&\quad{\rm s.t.} \sum_{i=1}^N(d_{i,t+1}-c_{i,t+1}) = -\sum_{i=1}^N\widehat{AIE}_{i,t},\\
	&\qquad 0\leq c_{i,t+1}\leq\left(1-\delta_{i,t}\right)\overline c_i,\\
	&\qquad 0\leq d_{i,t+1}\leq\delta_{i,t}\overline d_i,\\
	&\qquad \underline x_i\leq x_{i,t}+\frac{\eta^\textrm c\tau}{E_i} c_{i,t+1} - \frac{\tau}{\eta^\textrm d E_i} d_{i,t+1}\leq\overline x_i,
\end{align}
\end{subequations}
where (\ref{problem}a) focuses on the real-time cost-effectiveness of AGC enhancement. Here we would like to state that problem (\ref{problem}) has to be solved in an online manner due to the lack of global information and explicit prediction models \cite{Oshnoei:2020}. For iteration $t$, each BESS has to first interact with the environment (i.e., implement control), and only after this can it observe the resultant cost through the rainflow-counting algorithm and access local and neighbors' AIE signals. In turn, the observations at iteration $t$ will be utilized to correct the previous decision at the next iteration $t+1$ so as to counteract the AIE better.
\begin{remark}
    We consider a two-stage market model generalized for pay-for-performance market designs \cite{8449100}. A participant is pre-paid in the first stage and will be penalized for being unable to fulfill the regulation requirement. Similar to the augmented Lagrangian formulation, whether or not including the regulation penalty in (\ref{problem}a) does not change the optimal solution. Thus, the regulation penalty and the constant payment are not presented in (\ref{cost_function}). 
\end{remark}

\section{Proposed Scheme}
\subsection{Distributed OCO}
First of all, for the sake of generality, we denote
\begin{subequations}
\begin{align}
	&u_{i,t+1}\coloneqq[d_{i,t+1},-c_{i,t+1}]^\top,\\ &h_{i,t}(u_{i,t+1})\coloneqq\textbf 1_2^\top u_{i,t+1}+\widehat{AIE}_{i,t},
\end{align}
\end{subequations}
where $\textbf 1_2 = (1,1)\in\mathbb R^2$. Replace (\ref{problem}c)--(\ref{problem}e) with projection operation which projects $u_{i,t+1}$ into its decision domain to meet the inequality constraints. The Lagrangian function is
\begin{align}\label{lagrangian}
	L_t(u_t,\lambda) = \sum\nolimits_{i=1}^{N}f_{i,t}(u_{i,t}) + \lambda\sum\nolimits_{i=1}^{N}h_{i,t}(u_{i,t}),
\end{align}
where $\lambda$ is the dual variable of problem (\ref{problem}). For convex-constrained optimization problems, under Slater’s condition, a necessary and sufficient condition for primal-dual optimality is being the saddle point of the Lagrangian:
\begin{align}
	\frac{\partial L_{t}}{\partial u_{i,t}} &= (\frac{\partial f_{i,t}}{\partial d_{i,t}}, -\frac{\partial f_{i,t}}{\partial c_{i,t}}) + \textbf 1_2\lambda,\label{eq20}\\
	\frac{\partial L_{t}}{\partial\lambda} &= \sum_{i=1}^N h_{i,t}(u_{i,t}).\label{eq21}
\end{align}

It is evident from the following formula that the gradients involve global information such as $\lambda$ and $\sum_{i=1}^N h_{i,t}(u_{i,t})$. Modifications to the Arrow-Hurwicz-Uzawa algorithm are required. Consider a two-way network for agent communication. Two agents are said to be neighboring if there exists a communication link between them. We introduce a matrix $W=[w_{ij}]$ to model the communication topology by setting $w_{ij}\in\mathbb R_{>0}$ for neighboring agents $i$,$j$ and $w_{ij}=0$ otherwise. Note that $W$ needs to be doubly-stochastic, that is to say, $\sum_{i=1}^N w_{ij} = \sum_{i=1}^N w_{ji} = 1$.

Two auxiliary variables are introduced as the local estimates of the global information $\lambda$ and $\sum_{i=1}^N h_{i,t}(u_i)$ for each agent
\begin{align}\label{eq28}
    \tilde\lambda_{i,t}\coloneqq \sum\nolimits_{j=1}^{N}w_{ij}\lambda_{j,t},\quad\tilde y_{i,t}\coloneqq\sum\nolimits_{j=1}^{N}w_{ij}y_{j,t},
\end{align}
which compute the weighted averaging of local and neighbors' information.

Then, a local estimation of $\frac{\partial f_{i,t}}{\partial u_{i,t}}$ is given by
\begin{align}
	s_{i,t} = (\frac{\partial f_{i,t}}{\partial d_{i,t}}, -\frac{\partial f_{i,t}}{\partial c_{i,t}}) + \textbf 1_2\tilde\lambda_{i,t}.\label{algorithm_s}
\end{align}
Based on the above, we present an optimization algorithm for ORRA, which is executed in a distributed, online fashion and summarized in Algorithm 1.

The proposed algorithm incorporates adaptive learning rates that incorporates a two-phase switching mechanism to adapt OCO to dynamic control problems. The learning rate $\eta_t$ that monotonically decreases within Phase 1 is introduced by virtue of the dual-bounded technique in \cite{Cao:2021}. The basic idea is to impede the growth of $\lambda_{i,t}$ by introducing an additional term $\eta_t\Vert\lambda_t\Vert^2$ to the Lagrangian function. As a result, $\lambda_{i,t}$ can be tightly bounded within a certain range, and more emphasis is placed on the level of constraint compliance, contributing to reduced constraint violations while preserving the optimality of the final results. This approach provides a specialized treatment for dynamic fit in response to load transients, as it accounts for the changing nature of the system \cite{9184135}. However, the algorithm in \cite{9184135} would virtually terminate as learning rates diminishing to zero monotonically, hence conflicting with the goal of dynamic control at fundamental aspects. On the other hand, \cite{zhao:2020} and \cite{Zhou:2020} compromise for a constant learning rate $\kappa\in\mathcal O_{+}(1/\sqrt \mathcal T)$ that may fail to fully exploit the fast-ramping characteristics of BESSs. To cater to the technical requirements of both OCO and dynamic control, we integrate the principles of both designs by considering adaptive learning rates $\kappa_t$ and $\eta_t$, for which a two-phase switch mechanism is introduced. When $t$ approaches the iteration threshold $\mathcal T$, the algorithm transitions from Phase 1 to Phase 2, as illustrated in Algorithm 1. If $\vert\Delta f_t\vert$ exceeds the frequency threshold ${\rm THR}_f$, it indicates the start of a new optimization stage which starts at Phase 1. The iteration count is reset and a larger step size is used to accelerate the ramp-up/down of BESSs.

\begin{figure}
	\removelatexerror
	\begin{algorithm}[H]
		\caption{Proposed Algorithm for ORRA}
		\LinesNumbered 
		\KwIn{Parameters $\alpha$, $\beta$, $\gamma$, $\kappa_0$, $\eta_0\in\mathbb R_{>0}$}
		Initialization: $\kappa_t=\kappa_0$, $\eta_t=\eta_0$, $u_{i,0}\in\Omega_{i,0}$, $\lambda_{i,0} = 0$, $y_{i,0} = \textbf{1}_2^\top u_{i,0}+\widehat{AIE}_{i,0}$\;
            Let $t\leftarrow 1$\;
		\While{\rm TRUE}
		{\eIf{$1\leq t<\mathcal T$}{Phase 1: $\kappa_t=\kappa_0t^{-\alpha}$, $\eta_t=\eta_0t^{-\beta}$\;}{Phase 2: $\kappa_t=\kappa_0\mathcal T^{-\alpha}$, $\eta_t=0$\;}{
			\For{$i=1,...,N$}{
                     Update RBF network if the conditions are met\;
                    Obtain online interpolant using current $\Delta f$ and calculate $\widehat{AIE}_{i,t}$\;
				Calculate $\tilde\lambda_{i,t}$, $\tilde y_{i,t}$, and $s_{i,t}$\;
				Update $u_{i,t+1}$ and $\lambda_{i,t+1}$:
				\begin{align}
					u_{i,t+1} &= \mathcal P_{\Omega_{i,t}}(u_{i,t}-\kappa_t s_{i,t});\label{algorithm_u}\\
					\lambda_{i,t+1} &= (1-\eta_t)\tilde\lambda_{i,t} + \gamma\kappa_t\tilde y_{i,t};\label{algorithm_lambda}
				\end{align}
			
				Incorporate the AIE signals:
				\begin{align}
				    \Delta h_{i,t} = \textbf 1_2^\top\left(u_{i,t+1}-u_{i,t}\right)\\
				   +\widehat{AIE}_{i,t}-\widehat{AIE}_{i,t-1};\notag
				\end{align}
				
				Update $y_{i,t+1}$:
				\begin{align}	
					y_{i,t+1} &=\tilde y_{i,t} + \Delta h_{i,t};\label{algorithm_y}
				\end{align}
				}
				\eIf{$\vert\Delta f_t\vert<{\rm THR}_{f}$}{Let $t \leftarrow t+1$\;}{Reset $t\leftarrow1$\;}
		}
		}
	\end{algorithm}
\end{figure}

\begin{remark}
    Projection operation $\mathcal P_{\Omega_{i,t}}$ in (\ref{algorithm_u}) is included to project decision variable $u_{i,t+1}$ into its domain $\Omega_{i,t}$. Parameters $\kappa_0,\eta_0,\gamma$ need to be tuned for satisfactory step sizes, and a careful balance is necessary from the convergence and stability perspective. Local information of all bus agents, namely $\lambda_{i,t}$ and $y_{i,t}$, are shared via the sparse communication network to steadily enhance ORRA's perception of global information per iteration. At steady-state, we have $\lambda_t\rightarrow\bar\lambda_t$ and $y_t\rightarrow\bar y_t$ (also, $\tilde\lambda_t\rightarrow\bar\lambda_t$ and $\tilde y_t\rightarrow\bar y_t$), where $\bar\lambda_t\coloneqq\textbf 1_N\sum_{i=1}^N \lambda_{i,t}/N$ and $\bar y_t\coloneqq\textbf 1_N\sum_{i=1}^N y_{i,t}/N$. As CGs slightly adjust their outputs to cover the net-load variation, the BESSs will gradually withdraw their contribution to AGC. This will ultimately lead to $d_t = \textbf 0_N$ and $c_t = \textbf 0_N$ if there are no further perturbations, which is the appearance of energy-neutral operation. Network constraints can also be taken into account when performing OCO \cite{Dall’Anese:2018,Patel:2019}.
\end{remark}

\subsection{Convergence Analysis}
Due to the fast-changing regulation requirement of ORRA, dynamic regret and dynamic fit are introduced to define its convergence. The dynamic regret is a performance metric computed for each iteration and summed up to measure how much the distributed solution deviates from the optimal trajectory from a centralized view.

We consider a complete optimization stage consisting of Phase 1 and Phase 2, with duration of $1<T<\mathcal T$ and $T'>1$, respectively. The dynamic regret is defined as
\begin{align}
	Reg(T+T') &= \sum_{t=1}^{T+T'}\sum_{i=1}^{N}f_{i,t}(u_{i,t})-\sum_{t=1}^{T+T'}\sum_{i=1}^{N}f_{i,t}(u_{i,t}^\star),\label{reg}
\end{align}
and the dynamic fit is introduced to quantify the overall constraint violations, which is the non-compliance with the regulation requirement
\begin{align}
	Fit(T+T') &= \sum_{t=1}^{T+T'}\sum_{i=1}^{N}h_{i,t}(u_{i,t}),\label{fit}
\end{align}
where $u_t^{\star}={\rm argmin}_{u_t\in\Omega_t}\sum_{i=1}^Nf_{i,t}(u_{i,t})$.
\begin{lemma}
    For Phase 1, the following inequality always holds
    \begin{align}\label{eq36}
        Reg_1(T)&\leq\sum_{t=1}^T\frac{1}{2\kappa_t}(\Vert u_t-u_t^\star\Vert^2-\Vert u_{t+1}-u_t^\star\Vert^2)\notag\\
        &\quad +\sum_{t=1}^T\frac{1}{2\gamma\eta_t}(\Vert\bar\lambda_t\Vert^2-\Vert\bar\lambda_{t+1}\Vert^2)\\
        &\quad +\sum_{t=1}^T\frac{\kappa_t}{2}\Vert s_t\Vert^2 + \sum_{t=1}^T\frac{1}{2\gamma\kappa_t}\Vert\gamma\kappa_t\tilde y_t-\eta_t\bar\lambda_t\Vert^2\notag\\
        &\quad +\sum_{t=1}^T\Vert\bar\lambda_t\Vert\cdot\Vert\tilde y_t-\bar y_t\Vert +\sum_{t=1}^T 2\Vert u_t\Vert\cdot\Vert\tilde\lambda_t-\bar\lambda_t\Vert\notag.
    \end{align}
\end{lemma}
\begin{proof}
	The proof of Lemma 1 is provided in Appendix.A.
\end{proof}
\begin{lemma}
    Let learning rate $\kappa_t\in\mathbb R_{>0}$ and $T>1$. Denote $S(T)\coloneqq\sum_{t=1}^T(\Vert u_t-u_t^\star\Vert^2-\Vert u_{t+1}-u_t^\star\Vert^2)/(2\kappa_t)$. Denote the bound on decision variables as $B_u$, where $B_u = \max(\overline d_i,\overline c_i,\forall i\in1,..,N)$. For Phase 1, the following statement is true if and only if $\kappa_t$ decreases with $t$
    \begin{align}\label{eq39}
        S(T)\leq NB_u^2/\kappa_T + NB_u V(T).
    \end{align}
\end{lemma}
\begin{proof}
	The proof of Lemma 2 is provided in Appendix.B.
\end{proof}

All these suggest that the boundedness of $Reg_1(T)$ relies on a sequence of results and the selection of learning rates. Note that the instantaneous dynamic regret $\sum_{i=1}^{N}f_{i,t}(u_{i,t})-\sum_{i=1}^{N}f_{i,t}(u_{i,t}^\star)$ may not perfectly converge to the exact value of zero. However, the algorithm provides a near-optimal allocation and meets the constraints in most circumstances. The following assumptions are required to facilitate the derivation of our main results.
\begin{assumption}
(1) The local cost functions $f_{i,t}:\mathbb R^2\rightarrow\mathbb R$ are Lipschitz continuous and there exists a positive constant $C_f$ such that $\Vert\partial f_{i,t}(x)\Vert\leq C_f$ for $\forall i\in1,...,N$ and $\forall t\in0,...,T-1$; (2) The time-varying disturbances the interconnected power system is subject to is norm-bounded.
\end{assumption}
\begin{remark}
This remark gives some important results for deriving the convergence analysis. Under Assumption 1.2, there exists a constant $B_y>0$ such that $\Vert y_{i,t}\Vert$ and $\Vert\tilde y_{i,t}\Vert$ are both uniformly bounded by $B_y$. When digging into the updating law (\ref{algorithm_lambda}), we have $\Vert\lambda_{i,t+1}\Vert=\Vert(1-\eta_t)\tilde\lambda_{i,t} + \gamma\tilde y_{i,t}\Vert\leq(1-\eta_t)\Vert\tilde\lambda_{i,t}\Vert + \gamma B_y$. According to (\ref{eq28}) and $\sum_{j=1}^N w_{ij}=1$, one might expect $\Vert\tilde\lambda_{i,t+1}\Vert= \Vert\sum\nolimits_{j=1}^N w_{ij}\lambda_{j,t+1}\Vert\leq\max(\Vert\lambda_{i,t+1}\Vert,\forall i\in1,...,N)$. It can be easily verified by mathematical induction that $\Vert\lambda_{i,t}\Vert,\ \Vert\tilde\lambda_{i,t}\Vert,\ \Vert\bar\lambda_{i,t}\Vert\leq \gamma B_y\kappa_t/\eta_t$. Further we have $\Vert s_{i,t}\Vert\leq\Vert\partial f_{i,t}(u_{i,t})\Vert + \Vert\textbf 1_2\tilde\lambda_{i,t}\Vert\leq C_f + 2\gamma B_y\kappa_t/\eta_t$.
\end{remark}

\begin{theorem}
    Let $V(T)\coloneqq \sum_{t=1}^T\Vert u_{t+1}^{\star}-u_{t}^{\star}\Vert/\kappa_t$ and $0<\alpha<\beta<1$. Under Assumption 1, it always holds that 
    \begin{align}
        Reg_1(T)\in\mathcal O_{+}(T^{1+2\beta-3\alpha}) + \mathcal O_{+}(V_T).
    \end{align}
    For the case that $\mathcal O_{+}(V_T)<\mathcal O_{+}(T^{1+2\beta-3\alpha})$, we have also
    \begin{align}
        Fit_1(T) \in \mathcal O(T^{1-\frac{2\alpha-\beta}{2}})+ \mathcal O(T^{1-\frac{\beta-\alpha}{2}}).
    \end{align}
\end{theorem}

\begin{proof}
Below, we are in a position to ensure the boundedness of each term of (\ref{eq36}) by first identifying their asymptotic growth rates against $T$. Lemma 1 together with Assumption 1 lead to $\lim_{T\rightarrow\infty}{S(T)/T} = 0$. Now, the second term of (\ref{eq36}) can be obtained as
\begin{align}\label{eq40}
    &\sum_{t=1}^T\frac{1}{2\gamma\kappa_t}(\Vert\bar\lambda_t\Vert^2-\Vert\bar\lambda_{t+1}\Vert^2)\notag\\
    &\leq\frac{N}{2\gamma}\left[\sum_{t=2}^T(\frac{1}{\kappa_t}-\frac{1}{\kappa_{t-1}})\Vert\bar\lambda_{i,t}\Vert^2+\frac{1}{\kappa_1}\Vert\bar\lambda_{i,1}\Vert^2\right]\notag\\
    &<\frac{N}{2\gamma}\left[\sum_{t=2}^T(\frac{1}{\kappa_t}-\frac{1}{\kappa_{t-1}})+\frac{1}{\kappa_1}\right]\cdot(\frac{\gamma B_y\kappa_T}{\eta_T})^2\notag\\
    &=\frac{N\gamma B_y^2\kappa_T}{2\eta_T^2}\in\mathcal O_{+}(T^{2\beta-\alpha}).
\end{align}

By substituting $C_f+2\gamma B_y\kappa_t/\eta_t$ for $\Vert s_{i,t}\Vert$ according to the results of Remark 4, the third term of (\ref{eq36}) becomes
\begin{align}\label{eq42}
    &\sum_{t=1}^T\frac{\kappa_t}{2}\Vert s_t\Vert^2\leq\sum_{t=1}^T\frac{N}{2}\kappa_t(C_f + 2\gamma B_y\kappa_t/\eta_t)^2\notag\\
    &\leq\sum_{t=1}^T\frac{N\kappa_0}{2\eta_0^2}\left[C_f^2 t^{-\alpha} + 4\gamma C_fB_y t^{\beta-2\alpha} + 4\gamma^2B_y^2 t^{2\beta-3\alpha}\right]\notag\\
    &<\frac{N\gamma C_f^2\kappa_0}{2\eta_0^2}\int_{1}^{T} t^{-\alpha} dt + 2N\gamma C_fB_y\int_{1}^{T} t^{\beta-2\alpha} dt\notag\\
    &\quad +2N\gamma B_y^2\int_{1}^{T} t^{2\beta-3\alpha} dt + \text{const.}\notag\\
    &\in O_{+}(T^{1+2\beta-3\alpha}).
\end{align}

Similarly, the fourth term of (\ref{eq36}) can be transferred into
\begin{align}\label{eq41}
	&\sum_{t=1}^T\frac{1}{2\gamma\kappa_t}\Vert\gamma\kappa_t\tilde y_t-\eta_t\bar\lambda_t\Vert^2\notag\\
	&\leq\sum_{t=1}^T\frac{1}{2\gamma\kappa_t}\left(\gamma^2 B_y^2\Vert\kappa_t\Vert^2 + \Vert\eta_t\bar\lambda_t\Vert^2+2\gamma B_y\Vert\kappa_t\eta_t\bar\lambda_t\Vert\right)\notag\\
	&\leq2\gamma B_y^2\sum_{t=1}^T\kappa_t\notag\\
	&\in\mathcal O_{+}(T^{1-\alpha}).
\end{align}

Furthermore, let us apply $\sum_{t=1}^T\Vert\bar y_t-\tilde y_t\Vert\in  \mathcal O_{+}(T^{1+\beta-2\alpha})$ and $\sum_{t=1}^T\Vert\bar\lambda_t-\tilde\lambda_t\Vert\in\mathcal O_{+}(T^{1-\alpha})$ \cite{9184135}. Hence omitting the less significant terms associated with $T$ allows us to conclude
\begin{align}\label{eq43}
    &\sum_{t=1}^T\Vert\bar\lambda_t\Vert\cdot\Vert\bar y_t-\tilde y_t\Vert +\sum_{t=1}^T2\Vert u_t\Vert\cdot\Vert\bar\lambda_t-\tilde\lambda_t\Vert\notag\\
    &\leq\frac{N\gamma B_y\kappa_T}{\eta_T}\sum_{t=1}^T\Vert\bar y_t-\tilde y_t\Vert +2NB_u\sum_{t=1}^T\Vert\bar\lambda_t-\tilde\lambda_t\Vert\notag\\
    &\in\frac{N\gamma B_y\kappa_T}{\eta_T}\cdot\mathcal O_{+}(T^{1+\beta-2\alpha}) + 2NB_u\cdot\mathcal O_{+}(T^{1-\alpha})\notag\\
    &\in\mathcal O_{+}(T^{1+2\beta-3\alpha}).
\end{align}
Thus, $Reg_1(T) = \mathcal O_+(T^{2\beta-\alpha}) + \mathcal O_+(T^{1+2\beta-3\alpha}) + \mathcal O_+(T^{1-\alpha}) + \mathcal O_+(V_T) \in \mathcal O_+(T^{1+2\beta-3\alpha})$, which implies that sublinear dynamic regret can be achieved for $2\beta<3\alpha$. Relaxation of the results in Lemma 1 gives us
\begin{align}
     Reg_1(T) + \frac{Fit_1(T)^2}{2N^2\sum_{t=1}^T\eta_t/\kappa_t}\in\mathcal O_+(T^{1+2\beta-3\alpha}),
\end{align}
of which the second term can obtained by reserving $\lambda$ in the derivation of (\ref{eq45}). Based on Assumption 1.2, we can now obtain that
\begin{align}
     Fit_1(T)^2 &\leq2N^2\sum_{t=1}^T\frac{\eta_t}{\kappa_t}\left[\mathcal O_+(T^{1+2\beta-3\alpha}) - Reg_1(T)\right]\notag\\
     &\leq\frac{2N^2\eta_0}{\kappa_0}\sum_{t=1}^T t^{\alpha-\beta}\left[\mathcal O_+(T^{1+2\beta-3\alpha}) + 2NC_fB_uT\right]\notag\\
     &\leq \mathcal O_+(T^{2+\beta-2\alpha}) + \mathcal O_+(T^{2+\alpha-\beta}),
\end{align}
and then
\begin{align}\label{eq50}
     Fit_1(T) \in \mathcal O(T^{1-\frac{2\alpha-\beta}{2}})+ \mathcal O(T^{1-\frac{\beta-\alpha}{2}}),
\end{align}
The proof is complete.
\end{proof}

\begin{remark}
    Based on the results of Theorem 1, we know that a sublinear dynamic regret can be achieved for $0<1+2\beta-3\alpha<1$, while a sublinear dynamic fit can be achieved for $\alpha<\beta<2\alpha$. It is able to ensure both sublinear dynamic regret and fit by selecting $\alpha$ and $\beta$ according to $\alpha<\beta<\frac{3}{2}\alpha$. Further, we show that dynamic regret and fit for a complete optimization stage have sublinear guarantees.
From the definitions we have
\begin{align}
        Reg(T+T') &= Reg_1(T) + Reg_2(T'),\label{eq51}\\
        Fit(T+T') &= Fit_1(T) + Fit_2(T').\label{eq52}
\end{align}
For Phase 2, it has been established that $Reg_2(T')\in\mathcal O_{+}(T'^{1-\alpha})$ and $Fit_2(T')\in\mathcal O(T')$ (Theorem 8 \cite{mahdavi2012trading}). Invoking the criterion that $\alpha<\beta<\frac{3}{2}\alpha$ we can reduce (50) into $Fit_1(T) \in \mathcal O(T^{1-\frac{\beta-\alpha}{2}})$, which implies
\begin{align}
        Reg(T+T') &\in \mathcal O_{+}(T^{1+2\beta-3\alpha}) + \mathcal O_+(T'^{1-\alpha}),\\
        Fit(T+T') &\in \mathcal O(T^{1-\frac{\beta-\alpha}{2}}) + \mathcal O(T').
\end{align}
Their convergence rates with respect to $T+T'$ depend on both the number of $T$ and $T'$ and the selection of $\alpha$ and $\beta$. However, it always true that
\begin{align}
        Reg(T+T') &<\mathcal O_+((T+T')^{1-(3\alpha-2\beta)}),\\
        Fit(T+T') &< \mathcal O(T+T').
\end{align}
The proposed algorithm is subject to both sublinear dynamic regret and fit. It is clear from (55) and (56) that dynamic fit is gained by sacrificing the dynamic regret to an extent, which aligns with with the idea of adaptive rates originated in \cite{jenatton2016adaptive} and offers flexibility in its practical implementation in power systems by taking carefully the trade-off between dynamic regret and fit.
\end{remark}

\section{Simulation Study}
\subsection{Simulation Setup}
Simulations are conducted in MATLAB/Simulink environment on an IEEE-14 bus system and an IEEE-39 bus system modified for AGC studies. The power system dynamics are linearized around its nominal operating point and captured by a low-order system model, as depicted in Fig. 1. Each bus is assigned with a bus agent that manages information exchange and all computations including online learning and online optimization. The communication topology is defined as an undirected graph that can be rather flexible but should at least contain a path between any two agents. We choose $\alpha=0.3$ and $\beta=0.4$. The control interval governing the updates of ORRA is set to 0.1 seconds, to which the communication delay is negligible in relative. Each BESS has a capacity of 2 MWh, peak power of 1 MW, and charging/discharging efficiencies for 0.95. Their initial SoCs are arbitrarily chosen from [0.2, 0.8].

The initial evaluation point for the market interval, denoted by $p_0$, is assumed to be 0 for the following case studies. We use a decaying coefficient to determine the next evaluation points for the distance infill method. During real-time operation, data are gathered online if certain conditions are met. To simulate noise or random events, we corrupt the data sample $P_i^{\rm dr}$ by simply adding a small-amplitude Gaussian noise. Numerous techniques have been proposed to address non-Gaussian noise in stochastic systems. Of note are B-spline RBF networks, which employ a B-spline model to approximate the output probability density function. It has been shown in \cite{yin2019rbfnn} that nonlinear filters based on B-spline models are capable of smoothing out non-Gaussian noise. Other possible approaches include probability density function transformation \cite{zhang2021novel}. This technique involves utilizing kernel density estimation that projects sample data to a high-dimensional space, thereby providing enhanced comprehension of the behavior of non-Gaussian stochastic systems.

\subsection{Effectiveness Verification}
This case study is provided as a calibration to examine the efficacy of AIE and the main features of ORRA. Simulation is carried out on the IEEE 14-bus system, where we single out Area 1 as the research object and consider BESS participation. Five BESSs are installed across Area 1 and their communication topology is given in Fig. 6.
\begin{figure}[!t]
	\centering	\subfigure{\includegraphics[width=0.53\columnwidth]{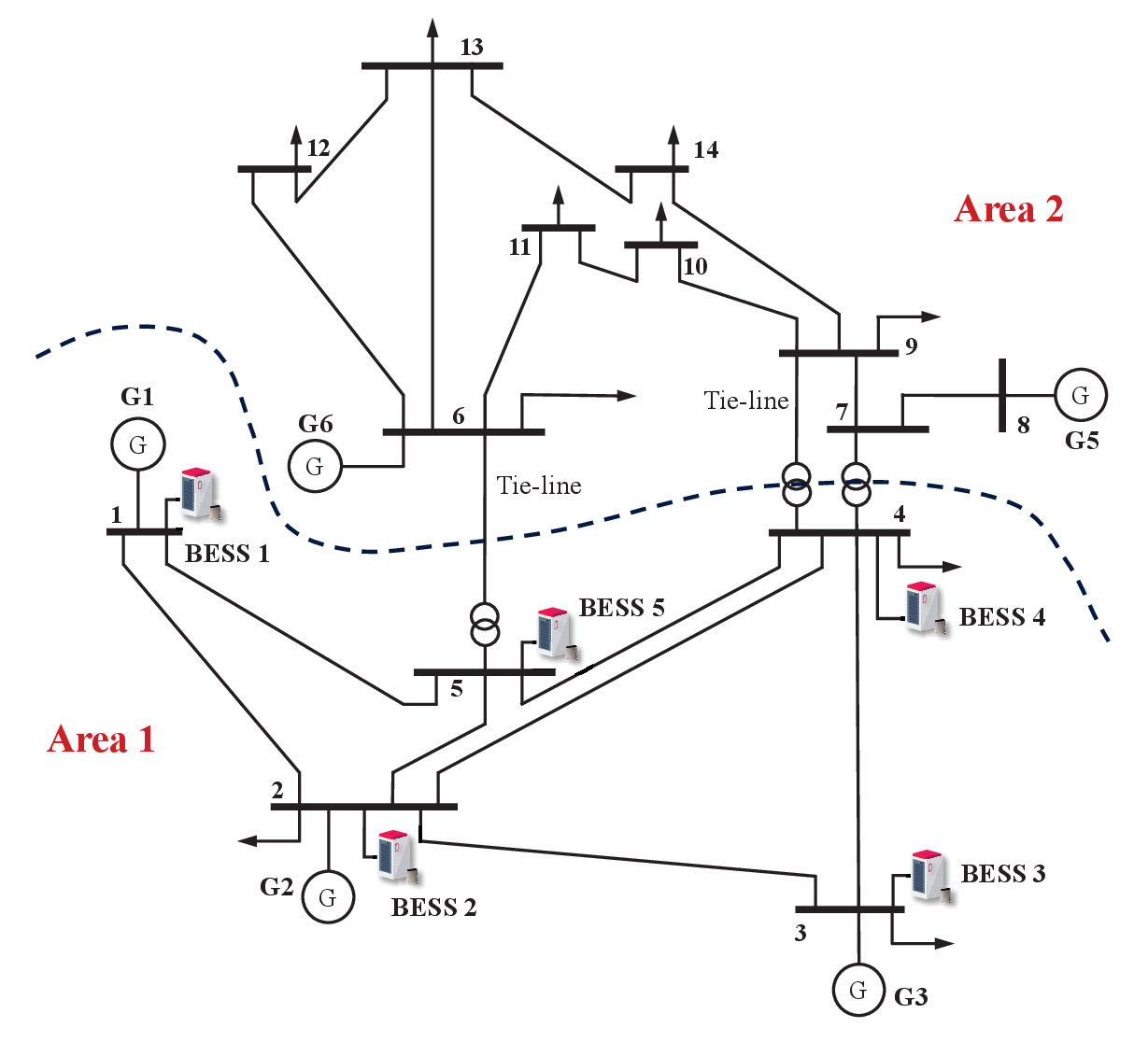}}
	\raisebox{0.45\height}{
	\subfigure{\includegraphics[width=0.42\columnwidth]{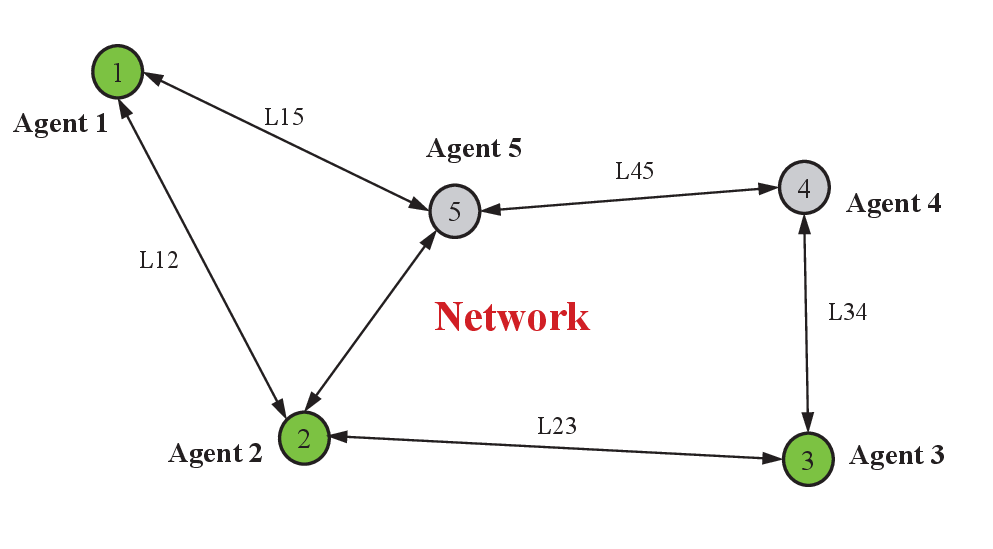}}
	}
	\caption{Single-line diagram of the two-area modified IEEE 14-bus system and communication topology of five BESSs in Area 1.}
\end{figure}

\subsubsection{Online Learning}
The learning results for a bus agent within an intraday market interval are shown in Fig. 7. The green line represents the actual $P$-$f$ characteristics, which is a high-level aggregation showing nonlinear frequency dependency, and the blue line represents the characteristics evaluated by the online interpolated RBF network. The online learning is a dynamic and adaptive process, as it involves collecting data in real-time and performing adjustment on the RBF network's weights and parameters.  According to the infill strategy, the new evaluation points $p_1$, $p_2$, $\dots$, $p_{10}$ are obtained in sequence over a relatively long time span. To highlight the model improvement, we split the learning process into two stages, represented in Fig. 7(a) and Fig. 7(b), respectively. During each stage, we obtain a sequence of evaluation points, including $p_1$, $p_2$, $\dots$, $p_{4}$ and $p_{5}$, $p_6$, $\dots$, $p_{10}$. Through a few evaluations that strike a balance between exploration and exploitation, it is evident that the interpolant of RBF network can well emulate the aggregated $P$-$f$ characteristics.
\begin{figure}[!b]
	\centering
	\subfigure[Stage 1.]{\includegraphics[width=0.45\columnwidth]{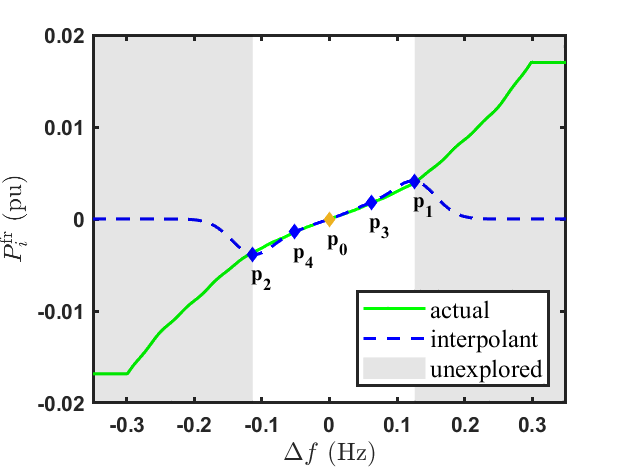}}
	\subfigure[Stage 2.]{\includegraphics[width=0.45\columnwidth]{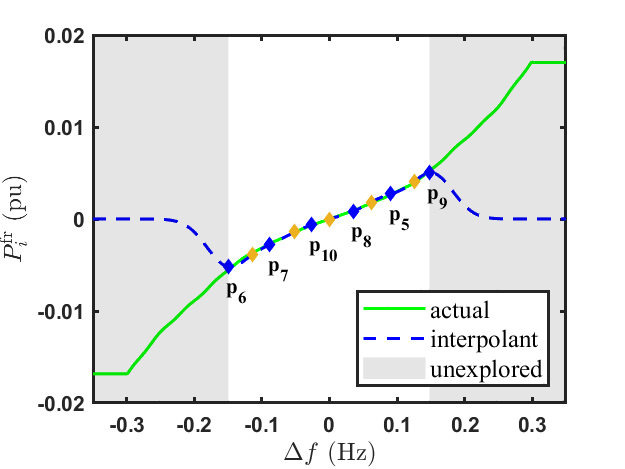}}
    \caption{Results of the online interpolated RBF network for bus agent 1.}
\end{figure}

\subsubsection{AGC Enhancement}
Fig. 8 shows the simulation results with respect to a step load increase of 5 MW at $t=10$s, which provides a visual assessment of the AGC enhancement. To avoid confusion, we provide an outline that compares four different configurations to highlight the significance of the AIE and BESS participation:
\begin{itemize}
    \item $\widehat{AIE}$+BESS: $\widehat{AIE}$-based AGC with BESS participation;
    \item $\widehat{AIE}$: $AIE$-based AGC without BESS participation
    \item $AIE$: $AIE$-based AGC without BESS participation;
    \item $ACE$: $ACE$-based AGC without BESS participation.
\end{itemize}
\begin{figure}[!t]
    \centering
	\subfigure[Area 1 frequency.]{\includegraphics[width=0.45\columnwidth]{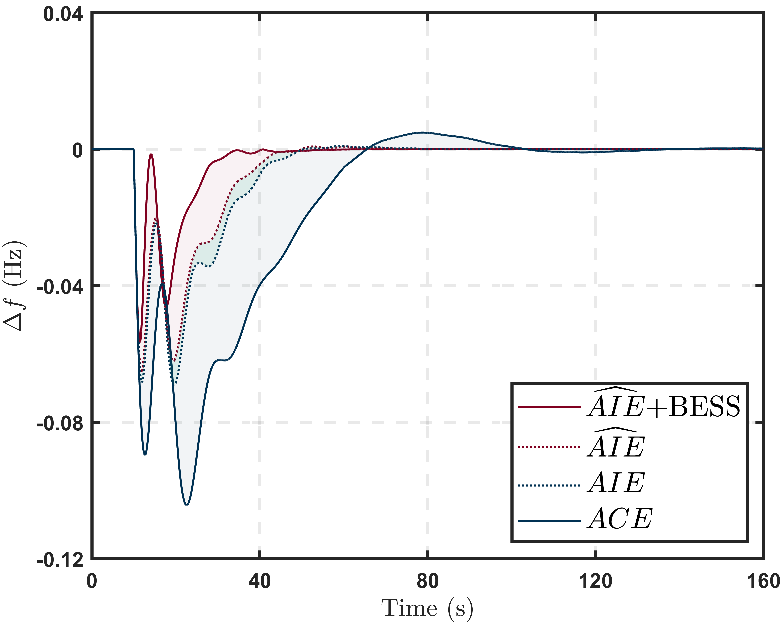}}
	\subfigure[Area 2 frequency.]{\includegraphics[width=0.45\columnwidth]{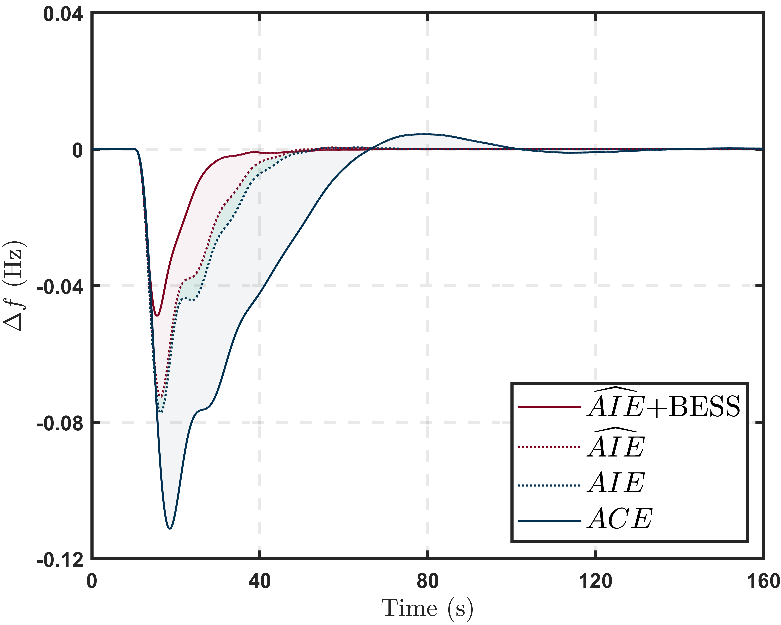}}
	\subfigure[BESS power.]{\includegraphics[width=0.45\columnwidth]{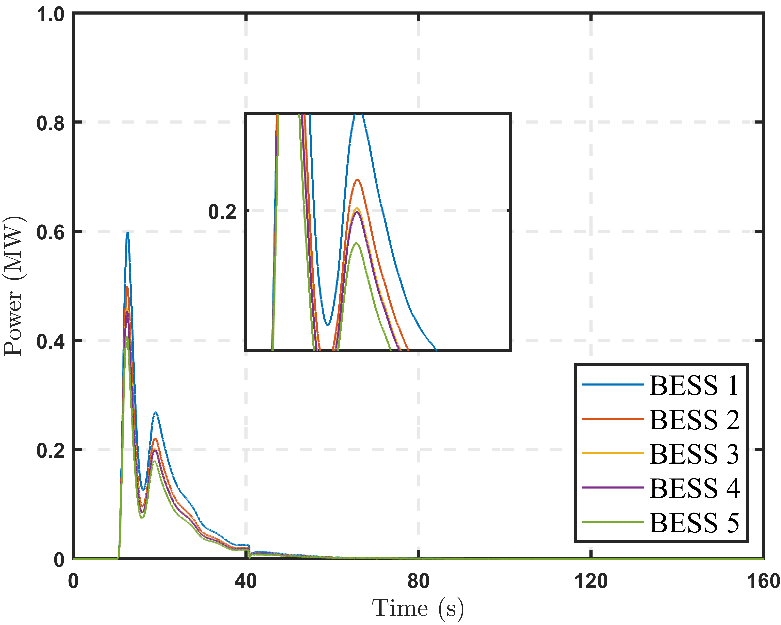}}
	\subfigure[BESS marginal cost.]{\includegraphics[width=0.45\columnwidth]{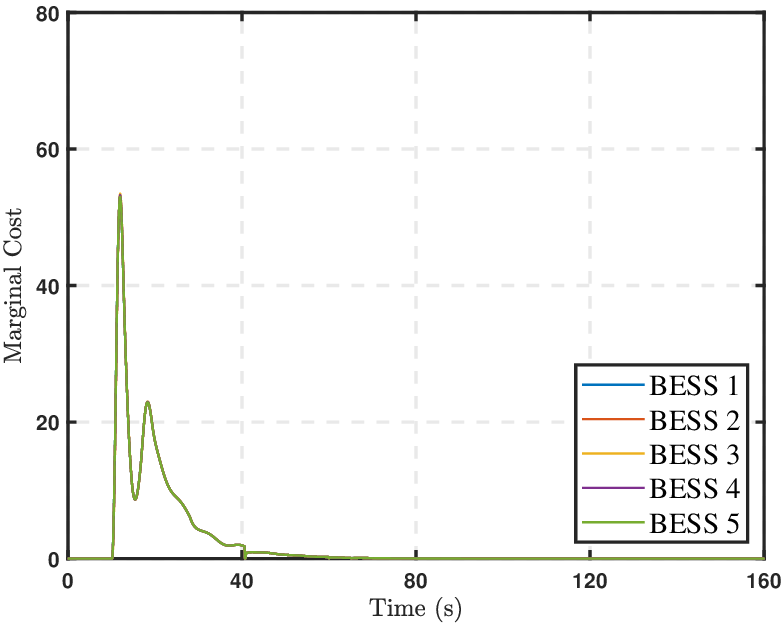}}
	\subfigure[Dynamic fit (iteration).]
	{\includegraphics[width=0.45\columnwidth]{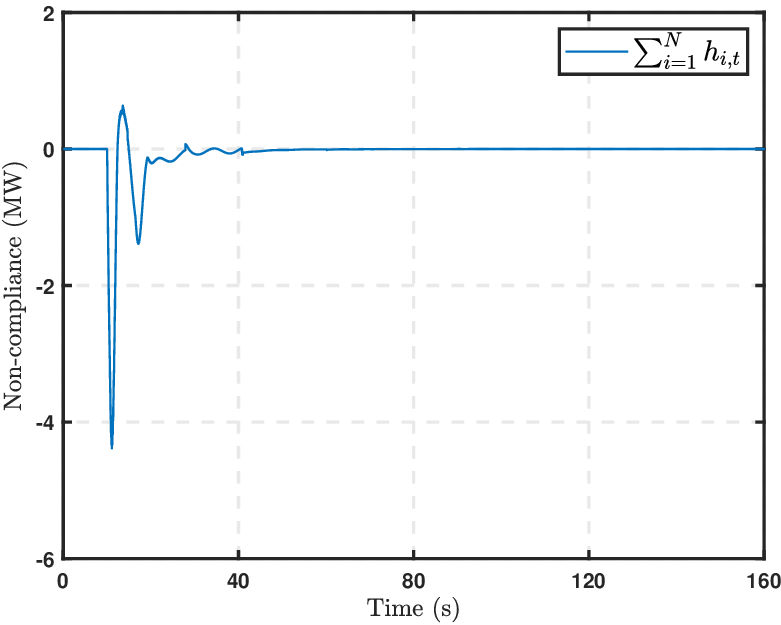}}
	\subfigure[Dynamic regret (iteration).]
	{\includegraphics[width=0.45\columnwidth]{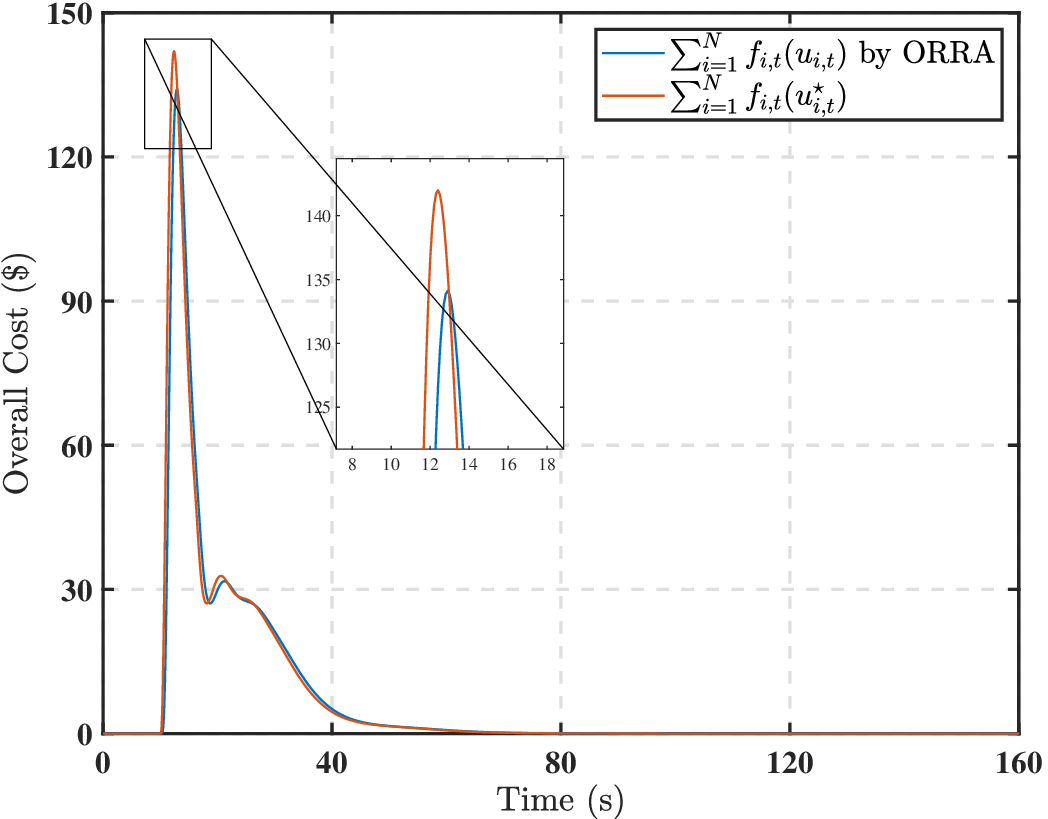}}
	\caption{Responses with respect to a step load increase of 5 MW in Area 1.}
\end{figure}
The frequency responses of the two areas are shown In Figs. 8(a)--(b). The AGC enhancement by the use of AIE and BESS participation are respectively marked in red and blue, whereas the green area shows the effects of incorporating the RBF network in the AIE. Benefiting from the capabilities of responding fast and precisely, the BESSs have reshaped the turbine-governor response and led to a significant enhancement in AGC performance, which is quantified by comparing the frequency drops of ``$\widehat{AIE}$+BESS" and ``$\widehat{AIE}$". As discussed in Section \uppercase\expandafter{\romannumeral3}. A, the ACE presumes a linear governor droop, ignoring the bias uncertainty raised by the slow and nonlinear dynamics of turbine-governor systems. Instead, the AIE enables a dynamic frequency bias by measuring the instantaneous difference between the governor input and the mechanical power output, thus mitigating the regulation inefficiencies. As shown in Figs. 8(c)--(d), BESSs respond to load transients at different levels due to cost heterogeneity and gradually detach from AGC in pace with the minimization of AIE.  Fig. 8(e) depicts the dynamic fit per iteration ($\sum_{i=1}^Nh_{i,t}(u_{i,t+1})$) which measures non-compliance with the ramping needs and is swiftly restored to a low level. In Fig. 8(f), the overall cost scaled to an hour is also compared with a centralized counterpart having full access to global information. The slight inconsistency between $\sum_{i=1}^Nf_{i,t}(u_{i,t})$ and $\sum_{i=1}^Nf_{i,t}(u_{i,t}^\star)$ reflects the dynamic regret per iteration.

\subsubsection{Long-Term Operation}
Following is an examination for integrity of grid ramp support in long-term operation. we introduce successive step changes as represented by the green line in Fig. 9(a), with positive values indicating power deficiency and negative values vice-versus. As a result, the BESSs shift between discharging mode and charging mode to counteract the AIE. The benefit of ORRA is characterized by a high degree of synergy, which is evident from the almost full complement provided by these two classes of regulation resources. This complementarity is particularly pronounced when the ramping capabilities of the CGs are insufficient during transients. Fig. 9(b) shows the evolution of the SoC levels over a time span of more than 30 minutes. The SoC levels remain closely around their initial values, avoiding continuous charging/discharging, which is a desirable property for maintaining the long-term operational integrity of BESSs.

\begin{figure}[!t]
	\centering
    \subfigure[CG power, BESS power, net-load variation (total).]{\includegraphics[width=0.85\columnwidth]{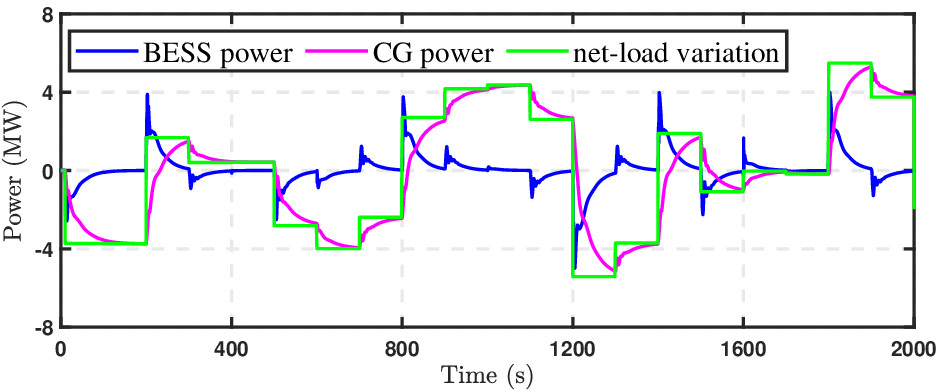}}
	\subfigure[Evolution of SoC levels across 30 minutes.]{\includegraphics[width=0.85\columnwidth]{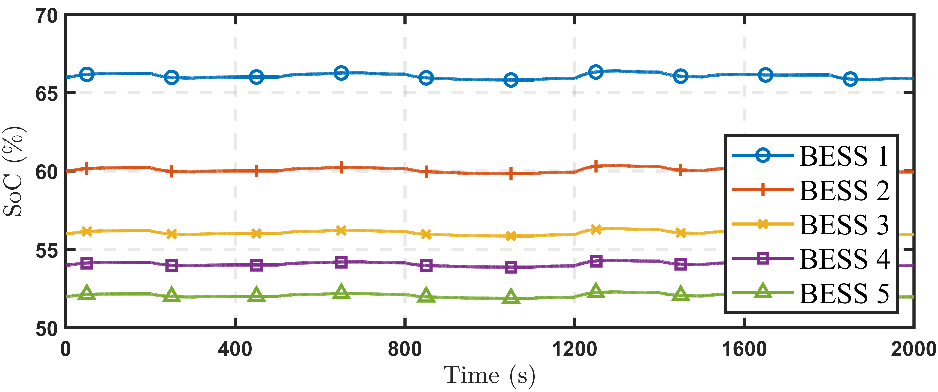}}
	\caption{Grid ramp support in long-term operation.}
\end{figure}

\subsection{Scalability Test}
This case study involves a scalability test on a two-area IEEE 39-bus system. The network topology is schematically illustrated in Fig. 10. It is worth mentioning that both areas are subjected to ORRA, which is also a necessary supplement to the previous case study to more rigorously verify its efficacy. We will show how ORRA performs in two areas of different sizes, subject to various load transients. Also, we pay particular attention to the impact of network topology. To this end, Area 1 adopts a linear topology, which has the minimum communication links among radial topologies. Conversely, we employ a mesh topology that is common for networked systems to Area 2. Two areas abide by the non-interaction principle of AGC.

\begin{figure}[htbp]
    \centering
    \includegraphics[width=0.55\columnwidth]{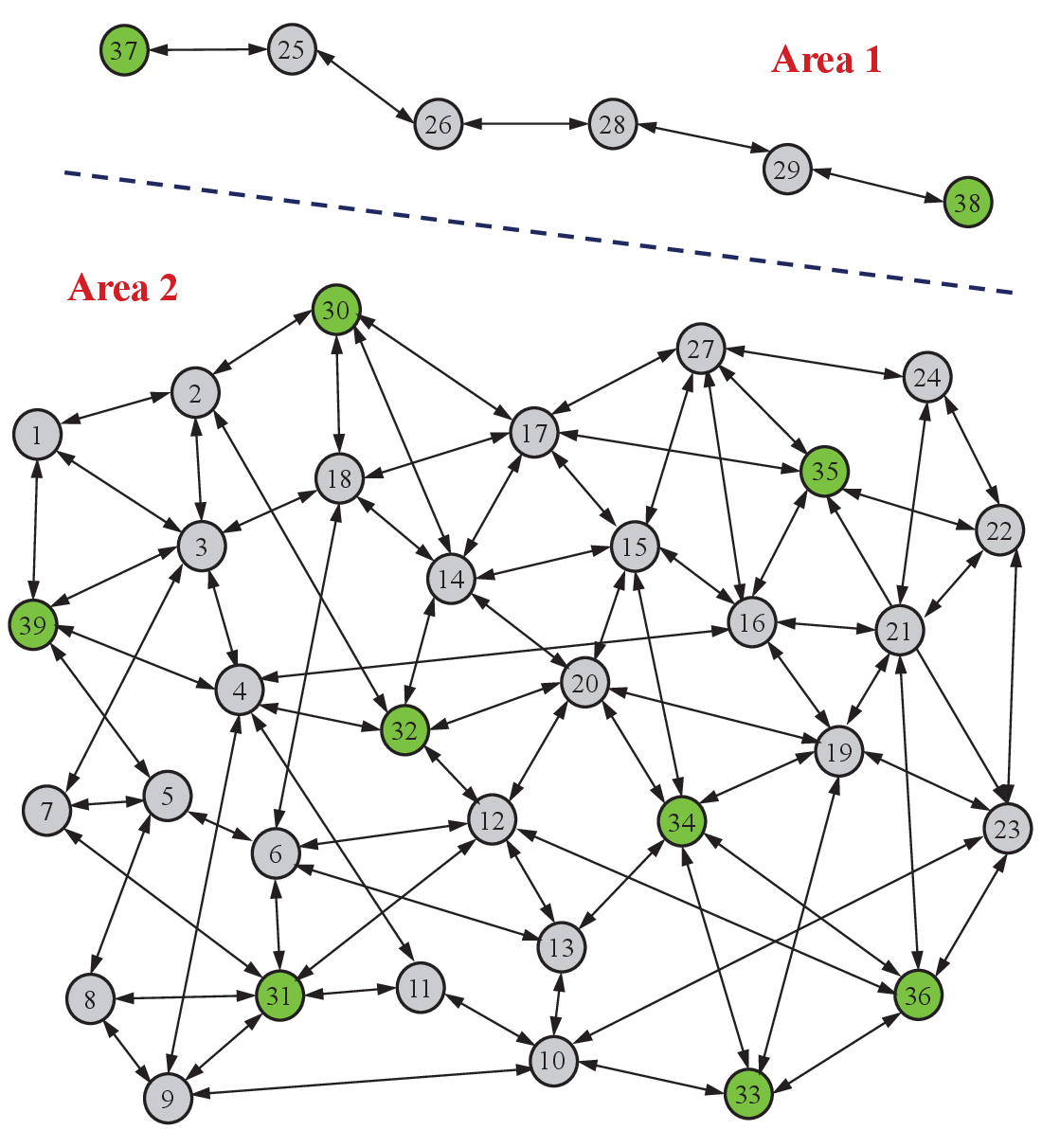}
	\caption{Communication network topology for the IEEE 39-bus system.}
\end{figure}

Fig. 11 presents the results of ORRA evaluated under different control intervals and compared with a Resource Allocation (RA) algorithm modified for ramping reserve allocation, referred to as RA here. Additionally, AIE-based AGC without BESS participation is provided as a benchmark. Comparing Figs. 11(b) and 11(d) shows that RA is more sensitive to the selection of control intervals and exhibits significant oscillations. On the other hand, ORRA, with its adaptive learning rate design and treatment on dynamic fit, offers superior frequency stabilization, as reflected in the reduced magnitudes and fewer oscillations of frequency deviation.
\begin{figure}[!b]
    \centering
    \subfigure[Area 1, $\tau=0.1$.]{\includegraphics[width=0.45\columnwidth]{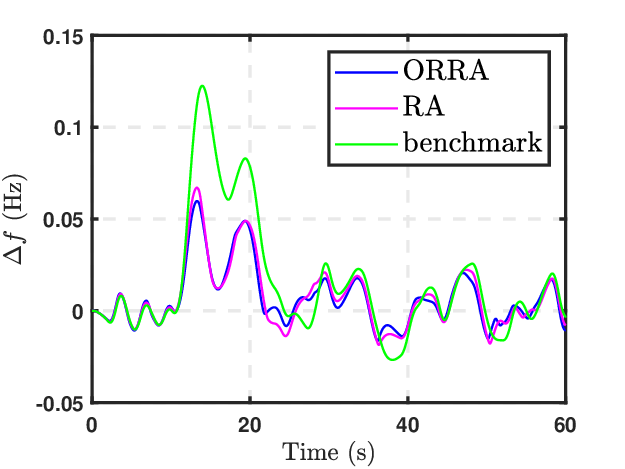}}
    \subfigure[Area 2, $\tau=0.1$.]{\includegraphics[width=0.45\columnwidth]{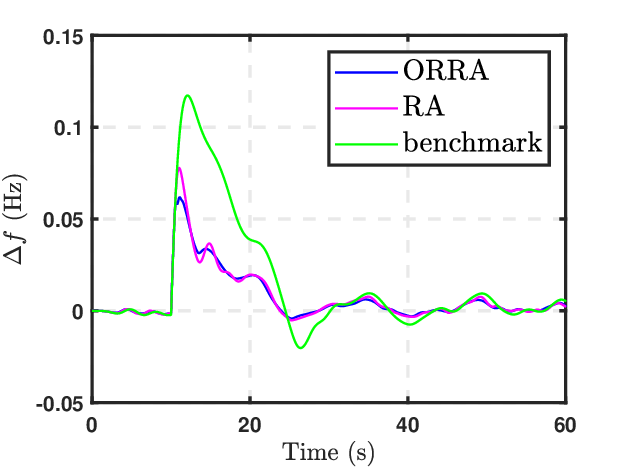}}
    \subfigure[Area 1, $\tau=0.2$.]{\includegraphics[width=0.45\columnwidth]{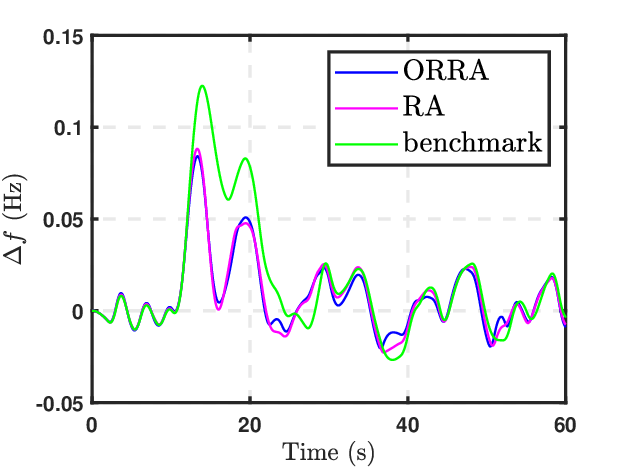}}
    \subfigure[Area 2, $\tau=0.2$.]{\includegraphics[width=0.45\columnwidth]{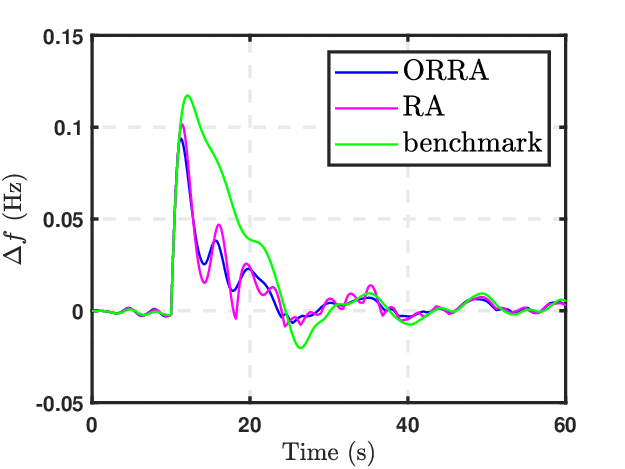}}
    \caption{Frequency responses under different control intervals. A 20 MW load decrease with small variations is introduced to Area 2.}
\end{figure}
\begin{figure}[tbp]
	\centering
	\subfigure[BESS Power, Area 1.]{\includegraphics[width=0.45\columnwidth]{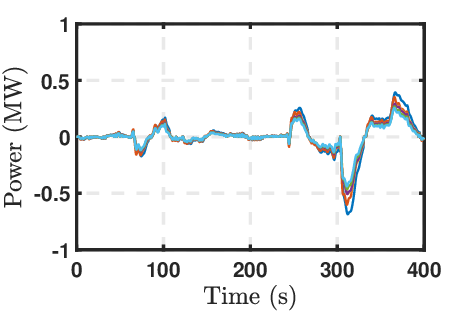}}
        \subfigure[BESS Power, Area 2.]{\includegraphics[width=0.45\columnwidth]{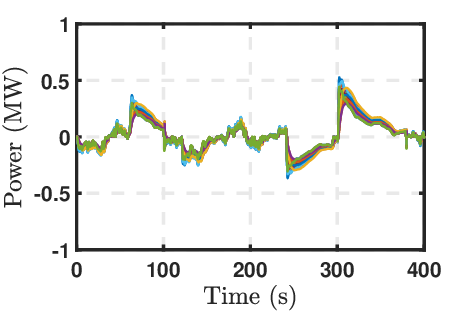}}
        \subfigure[Consensus of $\lambda$, Area 1.]{\includegraphics[width=0.45\columnwidth]{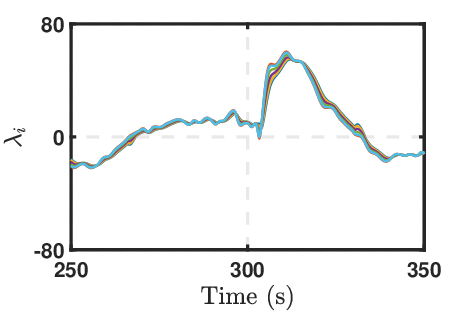}}
        \subfigure[Consensus of $\lambda$, Area 2.]{\includegraphics[width=0.45\columnwidth]{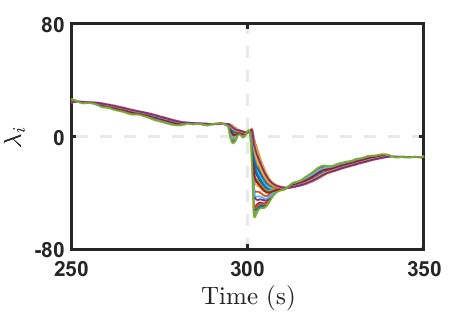}}
        \subfigure[Dynamic fit (iteration) under net-load variations, Area 2.]{\includegraphics[width=0.98\columnwidth]{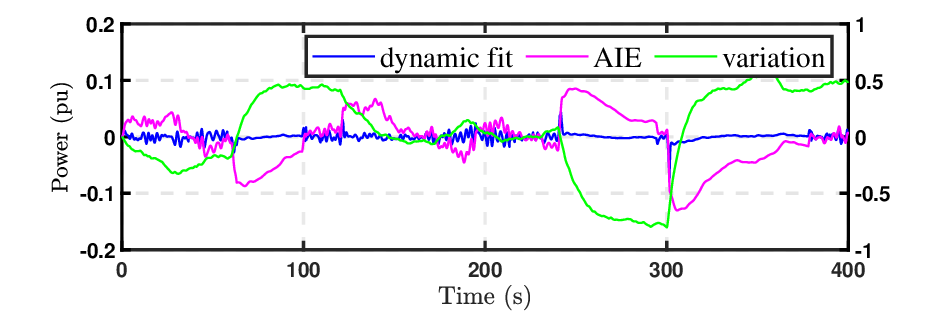}}
	\caption{Key results of the scalability test, where severe net-load variations bounded by [-1,1] pu are introduced to Area 2.}
\end{figure}

It can be observed from Figs. 12(a)-(b), the BESSs in the two different areas exhibit different directions of response with respect to net-load variations in Area 2. This is attributed to tie-line power flows,  which are also part of the AIE signal for each area and thereby cause inter-area oscillations that the BESSs are responsible for mitigating. While the ORRA algorithm can theoretically converge with any topology, provided the existence of a direct spanning tree between any two nodes, achieving the desired control performance may be more challenging with certain topologies. A linear topology can be acceptable for a small network, but for larger networks, merely ensuring a direct spanning tree is inadequate as the consensus of $\lambda$ is susceptible to network connectivity. Appropriate redundancy of the communication network is crucial not only for achieving the desired control performance but also for providing robustness in the event of communication failures. The results depicted in Figs. 12(c)-(d) highlight the crucial role of network connectivity for the effective application of ORRA. To enhance network connectivity, additional communication links could be considered. By recalling that the dynamic fit corresponds to the level of non-compliance, we can use the ratio between the dynamic fit and the AIE to quantify the effectiveness of online optimization. As shown in Fig. 12(e), the dynamic fit of each iteration is kept at a relatively low level compared to the AIE, whereas the green line represents the net-load variations and corresponds to the right axis. For coordinating 6 and 33 BESSs, the average computation time for each iteration on a laptop with 8 Intel Core i5 processors running at 2.4 GHz is 16.2 and 21.7 ms, respectively. Therefore, the proposed scheme scales very well in terms of computational time each node requires only a low-cost computing unit.

\section{Conclusion}
This paper introduces a new scheme, called ORRA, which aims to coordinate multiple BESSs in AGC using a distributed and online approach. Incorporating an online learning paradigm, we have proposed a variant of the AIE that can enhance the performance of AGC even in the absence of BESSs. Next, we utilize the AIE to develop a distributed OCO algorithm with adaptive learning rates and a two-phase switch mechanism to make ORRA practically implementable. The proposed scheme has been shown to be able to improve the transient behavior of the AGC system in an unknown and variable environment while maintaining acceptable dynamic fit. Moreover, the scheme leverages the synergy between BESSs and conventional generators such that BESSs operate during transients only, achieving nearly energy-neutral operation. Future research directions include extending the algorithm to account for communication delays and mixed-integer programming problems.

\section{Appendix}
\subsection{Proof of Lemma 1}
\begin{proof}
    According to (\ref{lagrangian}) and $\sum_{i=1}^N h_{i,t}(u_{i,t}^\star) = 0$, we have that $Reg_1(T)\equiv\sum_{t=1}^T L_t(u_t,\textbf 0_N)-\sum_{t=1}^T L_t(u_t^\star,\bar\lambda_t)$, which allows us to rewrite the dynamic regret as
    \begin{align}\label{eq45}
    \begin{split}
        Reg_1(T) &= \sum_{t=1}^T\left[L_t(u_t,\textbf 0_N)-L_t(u_t,\bar\lambda_t)\right]\\
        &\quad +\sum_{t=1}^T\left[L_t(u_t,\bar\lambda_t)-L_t(u_t^\star,\bar\lambda_t)\right].
    \end{split}
    \end{align}
    
    To move forward, we need to obtain the upper bounds of $\sum_{t=1}^T[L_t(u_t,\textbf 0_N)-L_t(u_t,\bar\lambda_t)]$ and $\sum_{t=1}^T[L_t(u_t,\bar\lambda_t)-L_t(u_t^\star,\bar\lambda_t)]$. From updating law (\ref{algorithm_lambda}), we have
    \begin{align}\label{eq46}
        \Vert\bar\lambda_{t+1}\Vert^2 &= \Vert\bar\lambda_{t} + (\gamma\kappa_t\tilde y_t-\eta_t\bar\lambda_t)\Vert^2\notag\\
        &\leq\Vert\bar\lambda_{t}\Vert^2 + \Vert\gamma\kappa_t\tilde y_t-\eta_t\bar\lambda_t\Vert^2 + 2(\gamma\kappa_t\tilde y_t-\eta_t\bar\lambda_t)^\top\bar\lambda_t\notag\\
        &\leq\Vert\bar\lambda_{t}\Vert^2 + \Vert\gamma\kappa_t\tilde y_t-\eta_t\bar\lambda_t\Vert^2 + 2\gamma\kappa_t\tilde y_t^\top\bar\lambda_t.
    \end{align}
    Since $\tilde y_t^\top\bar\lambda_t = (\tilde y_t-\bar y_t)^\top\bar\lambda_t + \bar y_t^\top\bar\lambda_t$ and $\bar y_t^\top\bar\lambda_t = L_t(u_t,\bar\lambda_t)-L_t(u_t,\textbf 0_N)$, (\ref{eq46}) gives the result that the first term of (\ref{eq45}) satisfies
    \begin{align}\label{eq47}
    \begin{split}
        &L_t(u_t,\textbf 0_N)-L_t(u_t,\bar\lambda_t)\\
        &\leq\frac{1}{2\gamma\kappa_t}(\Vert\bar\lambda_t\Vert^2-\Vert\bar\lambda_{t+1}\Vert^2) + \frac{1}{2\gamma\kappa_t}\Vert\gamma\kappa_t\tilde y_t-\eta_t\bar\lambda_t\Vert^2\\
        &\quad+\Vert\bar\lambda_t\Vert\cdot\Vert\tilde y_t-\bar y_t\Vert.
    \end{split}
    \end{align}
    
    As the next step, recalling updating law (\ref{algorithm_u}) along the property possessed by projection mapping that $\Vert\mathcal P_{\Omega}(x)-\mathcal P_{\Omega}(y)\Vert\leq\Vert x-y\Vert$ yields
    \begin{align}\label{eq48}
    \begin{split}
        \Vert u_{t+1}-u_{t}^{\star}\Vert^2&\leq\Vert u_{t}-u_{t}^{\star}-\kappa_t s_{t}\Vert^2\\
        &\leq\Vert u_{t}-u_{t}^{\star}\Vert^2 + \Vert\kappa_t s_{t}\Vert^2\\
        &\quad -2\kappa_t s_{t}^\top(u_{t}-u_{t}^{\star}).
    \end{split}
    \end{align}
    
    By the first-order property of characterization of convex functions, we have $-2\kappa_t s_{t}^\top(u_{t}-u_{t}^{\star})\leq-2\kappa_t[f_{t}(u_{t})-f_{t}(u_{t}^{\star})+(\textbf 1_2\tilde\lambda_{t})^\top(u_{t}-u_{t}^{\star})]$. As a result of $f_{t}(u_{t})-f_{t}(u_{t}^{\star}) =L_t(u_t,\bar\lambda_t)-L_t(u_t^\star,\bar\lambda_t)$, we can further conclude that
    \begin{align}\label{eq49}
    \begin{split}
        &L_t(u_t,\bar\lambda_t)-L_t(u_t^\star,\bar\lambda_t)\\
        &\leq\frac{1}{2\kappa_t}(\Vert u_t-u_t^\star\Vert^2-\Vert u_{t+1}-u_t^\star\Vert^2)+\frac{\kappa_t}{2}\Vert s_t\Vert^2\\
        &\quad +2\Vert u_t\Vert\cdot\Vert\tilde\lambda_t-\bar\lambda_t\Vert.
    \end{split}
    \end{align}
    Substituting (\ref{eq47}) and (\ref{eq49}) into (\ref{eq45}) and rearranging the terms ends the proof.
\end{proof}
\subsection{Proof of Lemma 2}
\begin{proof}
    We regroup $S(T)$ as the summation of $S_1(T)$ and $S_2(T)$ for notational simplicity, as shown by
    \begin{align}\label{eq50}
        S(T)&=\underbrace{\sum_{t=1}^T\frac{1}{2\kappa_t}\left(\Vert u_t-u_t^\star\Vert^2-\Vert u_{t+1}-u_{t+1}^\star\Vert^2\right)}_{S_1(T)}\\
        &\quad+\underbrace{\sum_{t=1}^T\frac{1}{2\kappa_t}\left(\Vert u_{t+1}-u_{t+1}^\star\Vert^2-\Vert u_{t+1}-u_t^\star\Vert^2\right)}_{S_2(T)}\notag.
    \end{align}
    By taking the similar approach alike (\ref{eq40}), $S_1(T)$ can be rearranged as
    \begin{align}\label{eq51}
        S_1(T) &=\frac{1}{2\kappa_1}\Vert u_1-u_1^{\star}\Vert-\frac{1}{2\kappa_{T+1}}\Vert u_{T+1}-u_{T+1}^{\star}\Vert\notag\\
        &\quad+\frac{1}{2}\sum_{t=2}^{T}(\frac{1}{\kappa_{t}}-\frac{1}{\kappa_{t-1}})\Vert u_{t}-u_{t}^{\star}\Vert^2\\
        &\leq NB_u^2/\kappa_T\notag.
    \end{align}
    From $\Vert x\Vert^2-\Vert y\Vert^2\leq\Vert x+y\Vert\cdot\Vert x-y\Vert$, it can be easily seen that
    \begin{align}\label{eq52}
        S_2(T)&\leq\sum_{t=1}^T\frac{1}{2\kappa_t}(2\Vert u_{t+1}\Vert+\Vert u_{t+1}^{\star}\Vert+\Vert u_{t}^{\star}\Vert)\cdot\Vert u_{t}^{\star}- u_{t+1}^{\star}\Vert\notag\\
        &\leq NB_uV(T).
    \end{align}
	Combining the results of (\ref{eq51}) and (\ref{eq52}) completes the proof.
\end{proof}

\ifCLASSOPTIONcaptionsoff
\newpage
\fi
\bibliography{ref}{}
\bibliographystyle{ieeetr}
\IEEEpubidadjcol
\newpage

\begin{IEEEbiographynophoto}{Yiqiao Xu} (Student Member, IEEE) holds a B.Eng. degree in detection, guidance, and control technology and an M.Sc. degree in advanced control and systems engineering, which were respectively received from Harbin Engineering University, Harbin, China and The University of Manchester, Manchester, U.K. in 2017 and 2018. He is currently pursuing the Ph.D. degree in Electrical and Electronic Engineering with The University of Manchester, Manchester, U.K., also working as a Research Associate there since late 2022. His research interests involve optimization and learning-based control, with a specific focus on smart grids and multi-energy systems. 
\end{IEEEbiographynophoto}

\begin{IEEEbiographynophoto}{Alessandra Parisio} (Senior Member, IEEE) is currently a Senior Lecturer with the Department of Electrical and Electronic Engineering, The University of
Manchester, Manchester, U.K., where she is the Principal or Co-investigator of research projects supported by EPSRC, Innovate U.K., EC H2020 totalling about £4 million as The University of Manchester share. Her main research interests include energy management systems under uncertainty, model predictive control, stochastic constrained control, and distributed optimisation for power systems. She is the Vice-Chair for Education of the IFAC Technical Committee 9.3. Control for Smart Cities, and the Editor of the Elsevier journal Sustainable Energy, Grids and Networks, Results in Control and Optimisation and the IEEE transactions on Automation and Science Engineering. She was the recipient of the IEEE PES Outstanding Engineer Award in January 2021 and the Energy and Buildings Best Paper Award for (for 2008-2017) in January 2019.
\end{IEEEbiographynophoto}

\begin{IEEEbiographynophoto}{Zhongguo Li} (Member, IEEE) received the B.Eng. and Ph.D. degrees in electrical and electronic engineering from The University of Manchester, Manchester, U.K., in 2017 and 2021, respectively. He was a Research Associate with the Department of Aeronautical and Automotive Engineering, Loughborough University, Loughborough, U.K., from 2020 to 2022. He is currently a Lecturer in Robotics and AI with Department of Computer Science, University College London, London, U.K. His research interests include optimization and decision-making for advanced control, game theory and learning in multi-agent systems, and their applications in autonomous vehicles and robotics.
\end{IEEEbiographynophoto}

\begin{IEEEbiographynophoto}{Zhen Dong} (Member, IEEE) received
the B.Eng. and M.Eng. degrees in electrical engineering from the Harbin Institute of Technology (HIT), Harbin, China, in 2016 and 2018, and the Ph.D. degree in electrical and electronic engineering from The University of Manchester, Manchester, U.K, in 2022. He is currently a Research Fellow with The University of Warwick, Coventry, U.K. His research interests include optimization and control of multiple-constraint electrical systems with applications to grid-connected wind farm management and high-performance motor control.
\end{IEEEbiographynophoto}

\begin{IEEEbiographynophoto}{Zhengtao Ding} (Senior Member, IEEE) received the B.Eng. degree from Tsinghua University, Beijing, China, in 1984, and the M.Sc.degree in systems and control, and the Ph.D. degree in control systems from The University of Manchester Institute of Science and Technology, Manchester, U.K. in 1986 and 1989, respectively. After working as a Lecturer with NgeeAnn Polytechnic, Singapore, for ten years, he joined The University of Manchester in 2003, where he is currently Professor of Control Systems with the Dept. of Electrical and Electronic Engineering. He is the author of the book: Nonlinear and Adaptive Control Systems (IET, 2013) and has published over 300 research articles. His research interests include nonlinear and adaptive control theory and their applications, more recently network based control, distributed optimization and distributed machine learning, with applications to power systems and robotics. Prof.Ding has served as an Associate Editor for the IEEE Transactions on Automatic Control, IEEE Control Systems Letters, and several other journals. He is a member of IEEE Technical Committee on Nonlinear Systems and Control, IEEE Technical Committee on Intelligent Control, and IFAC Technical Committee on Adaptive and Learning Systems.
\end{IEEEbiographynophoto}

\end{document}